\documentclass[letterpaper,12pt]{article}
\usepackage{amssymb}
\usepackage{amsmath}
\usepackage{amsfonts}
\usepackage{geometry,enumerate,bbm}
\usepackage{amsthm}
\usepackage{natbib}
\usepackage{epsfig}
\usepackage{float}
\usepackage[onehalfspacing]{setspace}
\usepackage{graphicx,color}
\usepackage{subcaption}
\usepackage{lscape}
\usepackage[flushleft]{threeparttable}
\usepackage{ulem}
\usepackage{scrextend}

\definecolor{dgreen}{rgb}{0.0, 0.5, 0.13}

\usepackage{pgfplots}

\usepackage[hyperindex,breaklinks]{hyperref}
\allowdisplaybreaks

\newtheorem{example}{Example}
\newtheorem{conjecture}{Conjecture}

\newtheorem{lemma}{Lemma}

\newcommand{\diag}{{\rm diag}}

\begin{document}

\title{\bf Approximate Functional Differencing\thanks{
We thank Stéphane Bonhomme for useful comments and discussions, and a referee for useful comments.  This research was
supported by the European Research Council grant  ERC-2018-CoG-819086-PANEDA and the Flemish Research Council grant G073620N.}}
\author{\setcounter{footnote}{2}Geert Dhaene\thanks{%
KU Leuven, \texttt{geert.dhaene@kuleuven.be} } \and Martin Weidner%
\thanks{%
University of Oxford, \texttt{martin.weidner@economics.ox.ac.uk} } }
\date{May 2023}

\maketitle
\thispagestyle{empty}
\setcounter{page}{1}

\begin{abstract}
    \noindent 
Inference on common parameters in panel data models with individual-specific fixed effects is a classic example of \citeauthor{neyman1948consistent}'s \citeyearpar{neyman1948consistent} incidental parameter problem (IPP). One solution to this IPP is functional differencing (\citealt{bonhomme2012functional}), which works when the number of time periods $T$ is fixed (and may be small), but this solution is not applicable to all panel data models of interest. Another solution, which applies to a larger class of models, is ``large-$T$'' bias correction (pioneered by \citealt{hahn2002asymptotically} and \citealt{HahnNewey2004}), but this is only guaranteed to work well when $T$ is sufficiently large. This paper provides a unified approach that connects these two seemingly disparate solutions to the IPP. In doing so, we provide an approximate version of functional differencing, that is, an approximate solution to the IPP that is applicable to a large class of panel data models even when $T$ is relatively small.
\end{abstract}

\bigskip
\noindent
{\bf Keywords:} { Panel data, discrete choice, incidental parameters, bias correction, functional differencing}

\bigskip
\noindent
{\bf JEL classification code:} {  C23}

\newpage 

\section{Introduction}

Panel data offer the potential to account for unobserved heterogeneity,
typically through the inclusion of unit-specific parameters; 
see \cite{arellano2003discrete} and \cite{ArellanoBonhomme2011} for reviews. 
Nonlinear panel data models, however, remain challenging to estimate,
precisely because in many models the presence of unit-specific -- or “incidental” -- parameters makes
the maximum likelihood estimator (MLE) of the common parameters inconsistent when
the number of observations per unit, $T$, is finite (\citealt{neyman1948consistent}).
The failure of maximum likelihood has prompted two kinds of reactions.

One approach is to look for point-identifying moment conditions that are free of incidental parameters.
Such moment conditions can come from a conditional or a marginal likelihood (e.g., \citealt{rasch1960studies}, \citealt{lancaster2000incidental}), an invariant likelihood (\citealt{moreira2009maximum}), an integrated likelihood (\citealt{Lancaster2002}), functional differencing (\citealt{bonhomme2012functional}), 
or from some other reasoning to eliminate the incidental parameters, for example, differencing 
in linear dynamic models (e.g., \citealt{arellano1991}). This approach is 
usually model-specific and is “fixed-$T$”, i.e., it seeks consistent estimation when $T$ is
fixed (and usually small). However, point-identifying moment conditions for small $T$ may 
not exist because point identification simply may fail; see \cite{HonoreTamer06} and 
\cite{chamberlain2010binary} 
for examples.

The other main approach is motivated by “large-$T$” arguments and seeks to reduce the large-$T$ bias
of the MLE or of the likelihood function itself or its score function
(e.g.,  \citealt{hahn2002asymptotically},
    \citealt{alvarez2003time},  \citealt{HahnNewey2004}, \citealt{ArellanoBonhomme2009}, \citealt{bonhomme2015grouped}, \citealt{DhaeneJochmans2015},
    \citealt{arellano2016likelihood},
    \citealt{FernandezValWeidner2016}). This approach is less
model-specific and may also be applied to models where point identification fails for small $T$.

The functional differencing method of  \cite{bonhomme2012functional} provides an algebraic approach 
to systematically find valid moment conditions in panel models with incidental parameters---if such moment
conditions exist. 
Related ideas are used in
\cite{Honore92}, \cite{Hu2002}, \cite{johnson2004identification}, 
\cite{kitazawa2013exploration},
\cite{honore2020moment}, \cite{honore2021dynamic}, and \cite{davezies2022fixed}.
In this paper, we extend the scope of functional differencing to models where point identification
may fail. In such models, exact functional differencing (as in \citealt{bonhomme2012functional})
is not possible,  but an approximate version thereof  yields moment conditions that 
are free of incidental parameters and that are approximately valid in the sense that their solution
yields a point close to the true common parameter value. Bonhomme’s method relies on the existence
of (one or more) zero eigenvalues of  a matrix of posterior predictive probabilities (or a posterior predictive density function) defined by the model. 
Our extension considers the case where all eigenvalues are positive and, therefore, point
identification fails, but where some eigenvalues are very close to zero. This occurs as the
number of support points of the outcome variable increases. Eigenvalues close to zero then lead
to approximate moment conditions obtained as a bias correction of an initially chosen
moment condition. The bias correction can be iterated, possibly infinitely many times. In 
point-identified models, the infinitely iterated bias correction is equivalent to functional
differencing. Therefore, approximate functional differencing can be viewed as finite-$T$ inference
in point-identified models, and as a large-$T$ iterative bias correction method
in models that are not point-identified.

The construction of approximate moment conditions is our main focus. Once such moment conditions are found, estimation follows easily using the (generalized) method of moments, and  the discussion of estimation is therefore deferred to later parts of the paper (from Section~\ref{sec:Estimation}).

We illustrate approximate functional differencing in a probit binary choice model. The implementation, 
including the iteration, is straightforward, only requiring elementary matrix operations. 
Indeed, one of the contributions of this paper is to show
how to iterate score-based bias correction methods for discrete choice
panel data relatively efficiently.

After introducing the model setup in Section~\ref{sec:Setup},
we review the main ideas behind functional differencing
in Section~\ref{sec:MainIdeas}. In Section~\ref{sec:AFD}
we introduce our novel bias corrections and explain how 
they relate to functional differencing.
 Section~\ref{sec:Numerical1} examines the eigenvalues of the matrix of posterior predictive probabilities in a numerical example. 
Section~\ref{sec:Estimation} briefly discusses estimation.
Further numerical illustration of the methods and some Monte
Carlo simulation results are  presented in Section~\ref{sec:Numerical2}. 
Section~\ref{sec:Extensions} discusses some extensions, in particular,
a generalization of the estimation method to average effects.
Finally, we provide some concluding remarks in Section~\ref{sec:conclusions}.

 \label{sec:Extensions}

\section{Setup}
\label{sec:Setup}

We observe outcomes $Y_i \in {\cal Y}$ and covariates $X_i \in {\cal X}$ for units $i=1,\ldots,n$. 
We only consider finite outcome sets $ {\cal Y}$ in this paper, but in principle all our results can be generalized to infinite   sets $ {\cal Y}$.
There are also  latent variables $A_i \in {\cal A}$, which are treated as nuisance parameters. We assume
that $(Y_i, X_i, A_i)$, $i=1,\ldots,n$, are independent and identical draws from a distribution with conditional outcome probabilities
\begin{align}
   {\rm Pr}\left( Y_i= y_i \, \big| \, X_i=x_i, \, A_i = \alpha_i \right) 
 =  f\left(y_i \, \big| \, x_i, \alpha_i, \theta_0 \right) ,
    \label{model}
\end{align}
where the function $ f\left(y_i \, \big| \, x_i, \alpha_i, \theta \right)$ is known (this function specifies ``the model''), but the true parameter value $\theta_0 \in \Theta \subset \mathbb{R}^{d_\theta}$ is unknown. Our primary goal 
in this paper is inference on $\theta_0$.

Let $\pi_0(\alpha_i \,| \,x_i)$ be the true distribution of $A_i$ conditional on $X_i=x_i$. Then,
the conditional distribution that can be identified from the data is
\begin{align}
   {\rm Pr}\left( Y_i= y_i \, \big| \, X_i=x_i  \right) 
 = \int_{\cal A} \,  f\left(y_i \, \big| \, x_i, \alpha_i, \theta_0 \right) \, \pi_0(\alpha_i \,| \,x_i) \, {\rm d} \alpha_i \, .
     \label{model2}
\end{align}
No restrictions are imposed on $\pi_0(\alpha_i \,| \,x_i)$ nor on the marginal distribution of  $X_i$, that is, we have a semi-parametric model with  unknown parametric component $\theta_0$
and unknown nonparametric component  $\pi_0(\alpha_i \,| \,x_i)$.

The setup just described covers many nonlinear panel data models with fixed effects. 
There, we observe outcomes $Y_{it}$ and covariates $X_{it}$ for unit $i$ over time periods
$t=1,\ldots,T$. For static panel models we then set $Y_i = (Y_{i1},\ldots,Y_{iT})$ and
$X_i = (X_{i1},\ldots,X_{iT})$, and the model is typically specified as
$$
f\left(y_i \, \big| \, x_i, \alpha_i, \theta \right)
  =  \prod_{t=1}^T f_*\left(y_{it} \, \big| \, x_{it}, \alpha_i, \theta \right) ,
$$
where $f_*\left(y_{it} \, \big| \, x_{it}, \alpha_i, \theta_0 \right)
 =  {\rm Pr}\left( Y_{it}= y_{it} \, \big| \, X_{it}=x_{it}, \, A_i = \alpha_i \right).$
Here,  $f_*\left(y_{it} \, \big| \, x_{it}, \alpha_i, \theta \right)$ often depends on
$ x_{it}$, $\alpha_i$, $\theta$ only through a single index
$x_{it}' \theta + \alpha_i$, where $\theta$ is a regression coefficient vector of the same dimension
as $x_{it}$, and $\alpha_i \in \mathbb{R}$ is an individual-specific fixed effect.
Of course, $\theta$ may also contain additional parameters (e.g., the variance of the error
term in a Tobit model).

For dynamic nonlinear panel models, we usually have to model the dynamics explicitly.
For example, we may include a lagged dependent variable in the model. In that case, 
assuming that $Y_{it}$ at $t=0$ is observed, we have
$Y_i = (Y_{i1},\ldots,Y_{iT})$ and
$X_i = (Y_{i0},X_{i1},\ldots,X_{iT})$, and the model is usually specified as
$$
f\left(y_i \, \big| \, x_i, \alpha_i, \theta \right)
  =  \prod_{t=1}^T f_*\left(y_{it} \, \big| \,y_{i,t-1}, x_{it}, \alpha_i, \theta \right) ,
$$
where $ f_*\left(y_{it} \, \big| \,y_{i,t-1}, x_{it}, \alpha_i, \theta \right)
 =  {\rm Pr}\left( Y_{it}= y_{it} \, \big| \, Y_{i,t-1}= y_{i,t-1} , \, X_{it}=x_{it}, \, A_i = \alpha_i \right)$.
Here, the initial observation $Y_{i0}$ is included in the conditioning variable $X_i$. 
In this way, the setup in equations \eqref{model} and \eqref{model2} also
covers dynamic panel data models.

The setup may also be relevant for applications outside of standard panel data,
e.g., pseudo-panels, network models, or games. But one typically needs
$Y_i$ to be a vector of more than one outcome to learn anything about $\theta_0$ since, in most models, the value of $\alpha_i$ alone can fully fit any possible outcome value if there is only a single outcome per unit  (i.e., if the sample is purely cross-sectional).

The main insights of our paper are therefore applicable more broadly, but
our focus will be on panel data. In particular, 
the following static binary choice panel data model will be our running
example throughout the  paper.

\renewcommand{\theexample}{1\Alph{example}}
\begin{example}[\bf Static binary choice panel data model]
     \label{ex:StaticPanel}
     Consider a static panel data model with $Y_i = (Y_{i1},\ldots,Y_{iT})$ and $X_i = (X_{i1},\ldots,X_{iT})$ where the outcomes $Y_{it} \in \{0,1\}$ are generated by
     $$
          Y_{it} = \mathbbm{1}(X'_{it} \, \theta_0 + A_i  \geq U_{it}) 
     $$
     and the errors $U_{it}$ are independent of $X_i$ and $A_i \in \mathbb{R}$, and are i.i.d.\ across $i$ and $t$ with cdf $F(u)$. 
     This implies
    $$
     f\left(y_i \, \big| \, x_i, \alpha_i, \theta \right) = \prod_{t=1}^T  \, [1-F(x'_{it} \,\theta + \alpha_i)]^{1-y_{it}} \,  [F(x'_{it} \, \theta + \alpha_i)]^{y_{it}} .
    $$
    For the probit model, we have $F(u)=\Phi(u)$, where $\Phi$ is the standard normal cdf, and for the logistic model we have
    $F(u)=   (1+e^{-u})^{-1}$. 
    To make the example even more specific, we consider a single binary covariate $X_{it} \in \{0,1\}$ such that for all $i=1,\ldots,n$ we have
    \begin{align*}
        X_{it} = \mathbbm{1}( t > T_0 ) \, ,
    \end{align*}
    for some $T_0 \in \{1,\ldots,T-1\}$, that is, $ X_{it}$ is equal to zero 
    for the initial $T_0$ time periods, and is equal to one for the remaining
     $T_1 = T-T_0$ time periods.
     Here $T_0$, and therefore $X_i$, is non-random and constant across~$i$, so
     we can simply write 
$f\left(y_i \, \big| \, \alpha_i, \theta \right)$ instead of
$ f\left(y_i \, \big| \, x_i, \alpha_i, \theta \right)$.
     The parameter of interest, $\theta_0 \in \mathbb{R}$,
    is one-dimensional.
\end{example}

\begin{example}[\bf Example~\ref{ex:StaticPanel} reframed]
     \label{ex:StaticPanel2}
     Consider Example~\ref{ex:StaticPanel}, but denote the binary outcomes now as $Y_{it}^* \in \{0,1\}$ and define the outcome $Y_i$ for unit $i$ as the pair
     \begin{align}
         Y_i  = \left( Y_{i,0} , Y_{i,1} \right) :=  \left(  \sum_{t=1}^{T_0}  Y^*_{it}  , \,  \sum_{t=T_0+1}^{T}  Y^*_{it}  \right) \in  \{ 0 ,\ldots,T_0 \} \times \{0, \ldots,T_1\} = {\cal Y}.
         \label{MappingYtoY}
     \end{align}
     Here, $Y_{i,0} = \sum_{t=1}^T Y^*_{it} \, (1-X_{it})$ is  the number of outcomes for unit $i$
     for which $Y_{it}^*=1$ within those time periods that have $X_{it}=0$, while 
     $Y_{i,1} = \sum_{t=1}^T Y^*_{it} \, X_{it}$ is  the number of outcomes with $Y_{it}^*=1$ for the time periods with $X_{it}=1$.
     This implies that
     $$
          f\left(y_i \, \big| \, \alpha_i, \theta \right)  =   { T_0 \choose y_{i,0}  } 
           [1-F(\alpha_i)]^{T_0 - y_{i,0} } \,  [F( \alpha_i)]^{y_{i,0} }  { T_1 \choose y_{i,1}  }      [1-F(\theta+\alpha_i)]^{T_1 - y_{i,1} } \,  [F(\theta+ \alpha_i)]^{y_{i,1} },
     $$
     where we drop $x_i$ from $f\left(y_i \, \big| \, x_i, \alpha_i, \theta \right)$  
     since it is non-random and constant across $i$.
     The parameter of interest, $\theta_0 \in \mathbb{R}$, is unchanged.

\end{example}
\renewcommand{\theexample}{\arabic{example}}

From the perspective of parameter estimation, Example~\ref{ex:StaticPanel2} is 
completely equivalent to Example~\ref{ex:StaticPanel}, because
 $Y_i$ in Example~\ref{ex:StaticPanel2} is a minimal sufficient statistic for the parameters $(\theta_0,\alpha_i)$ in
Example~\ref{ex:StaticPanel}.
Nevertheless, the outcome space in Example~\ref{ex:StaticPanel} is larger
($|{\cal Y}|=2^T$) than the  outcome space in Example~\ref{ex:StaticPanel2}
($|{\cal Y}|=(T_0+1)(T_1+1)$), and this will make a difference in our discussion
of moment conditions in these two examples below.

\section{Main idea behind functional differencing}
\label{sec:MainIdeas}

We now explain the main idea behind the
functional differencing method of \cite{bonhomme2012functional}. 
Our presentation is similar to that in \cite{honore2020moment}. However, our
goal here is much closer to that in Bonhomme's original paper
because we want to describe a general
estimation method, one that is applicable to a very large class of models, as opposed
to obtaining an analytical expression for moment conditions in specific models.

\subsection{Exact moment conditions}
\label{subsect:ExactMoments}

Consider the model described  by \eqref{model} and \eqref{model2}, where our goal is
to estimate $\theta_0$.
Functional differencing (\citealt{bonhomme2012functional}) aims to find moment functions ${\mathfrak m}(y_i,x_i,\theta) \in \mathbb{R}^{d_{\mathfrak m}}$ such that the model implies,
for all $x_i$ and $\alpha_i$, that
\begin{align}
     \mathbb{E}\left[ {\mathfrak m}(Y_i , X_i, \theta_0) \, \big| \, X_i=x_i, \, A_i = \alpha_i  \right] = 0 
     \label{moment}
\end{align}
or, equivalently,
\begin{align*}
     \sum_{y \in {\cal Y}} \, {\mathfrak m}(y , x_i, \theta) \,  f\left(y \, \big| \, x_i, \alpha_i, \theta \right)   = 0,
\end{align*}
since we want \eqref{moment} to hold for all possible $\theta_0 \in \Theta$.
Verifying that ${\mathfrak m}(y_i,x_i,\theta)$ satisfies this conditional moment condition only requires knowledge of the model $ f\left(y_i \, \big| \, x_i, \alpha_i, \theta \right) $,
not of the observed data. Note that 
${\mathfrak m}(y_i,x_i,\theta)$ does not depend on $\alpha_i$,
but nevertheless should have zero mean conditional on any realization
  $ A_i = \alpha_i$. This is a strong requirement, and we will get 
back to this below.

Once we have found such  valid moment functions ${\mathfrak m}(y_i,x_i,\theta)$, we can choose an arbitrary (matrix-valued) function $g(x_i,\theta) \in \mathbb{R}^{d_m \times d_{\mathfrak m}}$, and
define $$m(y_i,x_i,\theta) := g(x_i,\theta) \, {\mathfrak m}(y_i , x_i, \theta),$$
which is a vector of dimension $d_m$.
By the law of iterated expectations, we then obtain, under weak
regularity conditions,
the unconditional moment condition
\begin{align}
    \mathbb{E}\left[ m(Y_i,X_i,\theta_0)  \right] = 0,
     \label{moment2}
\end{align}
which we can use to estimate $\theta_0$ by the generalized method of moments (GMM, \citealt{hansen1982large}). The nuisance parameters $\alpha_i$ do not feature in the GMM estimation at all, that is, 
functional differencing provides a solution to the incidental
parameter problem (\citealt{neyman1948consistent}).

Of course, the key condition for consistent GMM estimation is that 
 $ \mathbb{E}\left[ m(Y_i,X_i,\theta)  \right] \neq 0$ for any $\theta \neq \theta_0$.
This identification condition is violated if ${\mathfrak m}(y_i , x_i, \theta)$ does
not depend on $\theta$ (a special case of which is 
${\mathfrak m}(y_i , x_i, \theta) = 0$, which is a trivial solution to \eqref{moment}).
Hence the moment functions must depend on $\theta$ to be informative about $\theta_0$.

\subsubsection*{Uninformative moment functions in Example~\ref{ex:StaticPanel}}

To give an example of a moment function that is uninformative about $\theta_0$, consider
Example~\ref{ex:StaticPanel}. 
Let $t$ and $s$ be two time periods where $X_{it}=X_{is}$.
Let $Y_{i,-(t,s)} \in \{0,1\}^{T-2}$ be the outcome vector $Y_{i}$ from which the outcomes $Y_{it}$  and $Y_{is}$ are dropped. Then, since $X_{it}=X_{is}$, the  
outcomes $Y_{it}$ and $Y_{is}$ are exchangeable and therefore
\begin{align*}
    \mathbb{E}\left( Y_{it} \, \big| \, Y_{i,-(t,s)} \right)
    =  \mathbb{E}\left( Y_{is} \, \big| \, Y_{i,-(t,s)} \right).
\end{align*}
This implies that for any function $g:  \{0,1\}^{T-2} \rightarrow \mathbb{R}$
the moment function
\begin{align}
     {\mathfrak m}(y_i,x_i,\theta) 
       := \left( y_{it} - y_{is} \right)  g(y_{i,-(t,s)})
     \label{uninformativeMoments}
\end{align}     
satisfies \eqref{moment}. This moment function does not depend on
$\theta$ and is therefore not useful for parameter estimation. (It is useful
for model specification testing, but we will not discuss this.)

Furthermore, one can show that every moment function
${\mathfrak m}(y_i,x_i,\theta)$ that satisfies \eqref{moment} in Example~\ref{ex:StaticPanel} is equal to a corresponding valid
moment function in Example~\ref{ex:StaticPanel2} plus a linear combination of moment functions
of the form \eqref{uninformativeMoments}.\footnote{
Let $y_i=y(y^*_i)$ be the mapping between an outcome $y^*_i$
in Example~\ref{ex:StaticPanel} and an outcome $y_i$ in Example~\ref{ex:StaticPanel2}, as defined by
\eqref{MappingYtoY}, and let $ {\cal Y}^*(y_i) = \left\{ y_i^* \, : \, y_i=y(y^*_i) \right\}$
be the set of outcomes $y^*_i$ that map to $y_i$.
Starting from a valid moment
function ${\mathfrak m}_A(y^*_i, \theta)$ in Example~\ref{ex:StaticPanel}
we obtain a valid moment function in Example~\ref{ex:StaticPanel2}
as ${\mathfrak m}_B(y_i, \theta) 
= \left| {\cal Y}^*(y_i) \right|^{-1} \sum_{y_i^* \in {\cal Y}^*(y_i)} {\mathfrak m}_A(y_i^* ,x_i, \theta)$. The null space of this linear mapping ${\mathfrak m}_A \mapsto {\mathfrak m}_B$
is spanned by moment functions of the form \eqref{uninformativeMoments}. 
This implies that the difference ${\mathfrak m}_A(y^*_i, \theta) - {\mathfrak m}_B(y(y^*_i)
  , \theta)$ is a linear combination of moment functions of the form \eqref{uninformativeMoments}.
}
Thus, from the perspective of constructing valid moment functions that
are  informative about $\theta_0$,  without loss of generality
we can focus on Example~\ref{ex:StaticPanel2} instead of 
Example~\ref{ex:StaticPanel}.
Example~\ref{ex:StaticPanel} is useful because it is a completely standard
panel model and it gives a simple example of valid moment functions 
that do not depend on $\theta$. From here onward, however, we will always use
Example~\ref{ex:StaticPanel2} as our running example.

\subsubsection*{Informative moment functions in Example~\ref{ex:StaticPanel2} for logistic errors}

Consider Example~\ref{ex:StaticPanel2} with logistic error distribution,
$F(u)=   (1+e^{-u})^{-1}$.
Then,  $Y_{i,0} + Y_{i,1}$ is a sufficient statistic for $A_i$, so
the distribution of $Y_i$ conditional on  $Y_{i,0} + Y_{i,1}$ does not depend
on $A_i$. It is well known that this implies that
the corresponding conditional MLE of $\theta_0$ is 
consistent as $n\to\infty$, for any fixed $T\ge 2$; see, e.g., \cite{rasch1960studies},
\cite{andersen1970asymptotic}, and
\cite{chamberlain1980analysis}.

Here, instead of considering 
conditional maximum likelihood, we focus purely on the existence of moment
conditions. Let $\bar y = (\bar y_0,\bar y_1) \in  \{ 0 ,\ldots,T_0 \} \times \{0, \ldots,T_1\}$
and $\tilde y = (\tilde y_0,\tilde y_1) \in  \{ 0 ,\ldots,T_0 \} \times \{0, \ldots,T_1\}$
be two possible realizations of $Y_i$ such that
 $\bar y_0+\bar y_1 = \tilde y_0 + \tilde y_1$ and  $\bar y \neq \tilde y$.
Since $Y_{i,0} + Y_{i,1}$ is a sufficient statistic for $A_i$, it must be the case
that the ratio  
$$
    r(\theta ) :=
   \frac{ f\left(\bar y \, \big| \, \alpha_i, \theta \right) }
   { f\left(\tilde y \, \big| \, \alpha_i, \theta \right) } 
$$
does not depend on $\alpha_i$. This implies that
\begin{align}
    {\mathfrak m}(y_i ,  \theta)
    := \mathbbm{1}\left\{ y_i = \bar y \right\}
    -  r(\theta ) \; \mathbbm{1}\left\{ y_i = \tilde y \right\}  
    \label{MomentFctSufficient}
\end{align}
satisfies  $ \mathbb{E}\left[ {\mathfrak m}(Y_i , \theta_0) \, \big| \, A_i=\alpha_i \right] = 0$.
A short calculation gives
\begin{align*}
    r(\theta )
    =   { T_0 \choose \bar y_{0}  }   { T_1 \choose \bar y_{1}  }  
    { T_0 \choose \tilde y_{0}  }^{-1}   { T_1 \choose \tilde y_{1}  }^{-1}
           \exp[(\bar y_1 - \tilde y_1)  \theta]  .
\end{align*}
Since we assume   $\bar y  \neq \tilde y$, 
the moment function $ {\mathfrak m}(y_i ,  \theta)$ indeed depends on $\theta$.
Furthermore, $ {\mathfrak m}(y_i ,  \theta)$ is strictly monotone in $\theta$
when $y_i=\Tilde{y}$ and constant in $\theta$ otherwise, and all outcomes 
are realized with positive probability. Hence
$\mathbb{E}\left[ {\mathfrak m}(Y_i,\theta)  \right] $
is strictly monotone in $\theta$ and  the condition
$ \mathbb{E}\left[ {\mathfrak m}(Y_i,\theta_0)  \right] = 0$
  uniquely identifies $\theta_0$.
  
The observation that the existence of a 
sufficient statistic for the nuisance parameter $A_i$ allows for identification
and estimation of $\theta_0$ is quite old (e.g., \citealt{rasch1960studies}).
However, the reason that the functional differencing method is truly powerful 
is that moment functions satisfying \eqref{moment} and identifying $\theta_0$ 
may exist even in models
where no sufficient statistic for $A_i$ is available. 
Examples of this are given by  
\cite{Honore92}, \cite{Hu2002}, \cite{johnson2004identification}, 
\cite{kitazawa2013exploration},
\cite{honore2020moment}, \cite{honore2021dynamic}, and \cite{davezies2022fixed}.
\cite{bonhomme2012functional} provides a computational method  for 
 obtaining moment functions ${\mathfrak m}(y_i,x_i,\theta) $ such that \eqref{moment} holds
 in a large class of models,
 while  \cite{honore2020moment} 
 discuss how to  obtain explicit algebraic formulas for moment conditions in specific models.
 \cite{dobronyi2021identification} show that additional 
 moment inequalities may exist that contain identifying information
 on $\theta_0$ that is not contained in the moment equalities.

Our example of a moment function in \eqref{MomentFctSufficient} is convenient and easy to understand, but it is not really
representative of the potential complexity of more general moment functions. The papers cited in the previous paragraph give a better view of the true capability of the functional differencing method in more challenging settings.

\subsection{Approximate moment conditions}

Functional differencing is a very powerful and useful method.
Nevertheless, there are many models to which it is not applicable.
The reason is that the condition in equation~\eqref{moment} is actually quite strong.
It requires us to find a function 
${\mathfrak m}(Y_i , X_i, \theta)$ that does not depend on $A_i$ at all,
but that is supposed to have a conditional mean of zero for any possible realization
of $A_i$. In most standard panel data models $A_i$ takes values in $\mathbb{R}$
($A_i$ can also be a vector), implying that
\eqref{moment} imposes an infinite number of linear restrictions. 

It is therefore perhaps unsurprising that there are many panel data models
for which \eqref{moment} has no non-trivial solution at all.
In Example~\ref{ex:StaticPanel2} we have shown the existence of
valid moment functions for the logit model,
 but it turns out that no valid moment function exists for
the probit model when $\theta_0 \neq 0$ (we have verified this non-existence numerically
for many values of $T$ and $T_0$).

Instead of trying to find moment functions satisfying \eqref{moment}, and hence \eqref{moment2},  exactly,
we  argue that it can also be fruitful to
search for moment functions that satisfy these conditions only approximately, i.e.,
\begin{align}
    \mathbb{E}\left[ m(Y_i,X_i,\theta_0)  \right] \approx 0 \, .
     \label{moment2approx}
\end{align}
For a given model $ f\left(y_i \, \big| \, x_i, \alpha_i, \theta \right) $
we might not be able to find
an exact solution to  \eqref{moment2}, but we might  be able to find 
a very good approximate solution.

Examples of approximate moment conditions are 
provided by the ``large-$T$'' panel data literature, which considers
asymptotic sequences where also $T \rightarrow \infty$
(jointly with $n \rightarrow \infty$). 
To illustrate the insights of this literature,
let $\widehat \alpha_i(\theta)$
be the MLE of $\alpha_i$ obtained
from  maximizing   $ f\left(Y_i \, \big| \, X_i, \alpha_i, \theta \right)$
over $\alpha_i\in\cal{A}$, and let
$\psi(y_i,x_i,\alpha_i,\theta)$ be a moment function that
satisfies $\mathbb{E}[\psi(Y_i,X_i,A_i,\theta_0)]=0$
in model \eqref{model}, e.g.,
$\psi(y_i,x_i,\alpha_i,\theta) =  \nabla_{\theta}
\left[\frac 1 T
\log  f\left(y_i \, \big| \, x_i, \alpha_i, \theta \right) \right]$.
Then, under standard regularity conditions, 
$\widehat \alpha_i(\theta_0)$ is consistent for $A_i$ as $T \rightarrow \infty$,
and a useful approximate moment function is therefore given by
$
    m(y_i,x_i,\theta) = 
    \psi(y_i,x_i,\widehat \alpha_i(\theta),\theta)
$.
In this example, the vague approximate statement in \eqref{moment2approx} can be made precise, namely one can show
that, as $T \rightarrow \infty$,
\begin{align}
    \mathbb{E}\left[ m(Y_i,X_i,\theta_0)  \right]  = O(T^{-1}) \, ,
         \label{moment2approxLargeT}
\end{align}
implying that the estimator of $\theta_0$ obtained from 
this approximate moment condition also has a bias of order $T^{-1}$.
It is possible to correct this bias
and obtain moment functions where the right hand
side of \eqref{moment2approxLargeT} is of order 
$T^{-2}$ or smaller, implying an even smaller bias for
estimates of $\theta_0$ when $T$ is sufficiently large;
see, e.g., \cite{HahnNewey2004},
\cite{arellano2007understanding,arellano2016likelihood},
\cite{ArellanoBonhomme2009},  and \cite{DhaeneJochmans2015}.
In this paper, we aim for even higher-order bias correction, where the remaining
bias is only of order $T^{-q}$, for some integer $q>0$, because we expect
better small-$T$ estimation properties from such higher-order corrections,
and it allows us to connect the bias correction results with the functional differencing method.
However, by correcting the bias in this way, one might very well increase the variance of the resulting estimator for $\theta_0$, as we will
see in our Monte Carlo simulations in Section~\ref{sec:Numerical2}.
The question of how to optimally trade off the bias vs.\ the variance  (using, e.g., a mean squared error
criterion function, as in \citealt{bonhomme2022minimizing}) is interesting and could lead to an optimal choice of the bias correction order $q$,
but we will not formalize this in the current paper.

\section{Approximate functional differencing}    
\label{sec:AFD}

In this section, we answer the following questions:
In a model where exactly valid moment functions as in \eqref{moment} do not exist,
is it still possible to construct useful moment functions $m(y_i,x_i,\theta)$
that are approximately valid as in \eqref{moment2approx}? And if yes, how can
we construct such moment functions in a principled way?

\subsection{Notation and preliminaries}   

Our discussion in this section is at the ``population level'', that is,
for one representative unit $i$. 
In the following, we therefore drop the index $i$ throughout. For example, 
instead of $Y_i$, $X_i$, $A_i$ we simply write $Y$, $X$, $A$.

\subsubsection{Prior distribution of the fixed effects}
Let $ \pi_{\rm prior}(\alpha\,|\,x) $ be a chosen prior distribution for $A$, conditional on $X=x$.
The prior should integrate to one, that is, $\int_{\cal A}   \, \pi_{\rm prior}(\alpha \,| \, x) \, {\rm d} \alpha = 1$, for all $x \in {\cal X}$, but
 we do not require $\pi_{\rm prior}$ to be identical to $\pi_0$.
We require non-zero prior probability for all points in the support of $A$, i.e.,
\begin{align}
     \pi_{\rm prior}\left(\alpha \, \,| \,x \right) > 0 , \qquad \text{for all   $\alpha \in {\cal A}$  and   $x \in {\cal X}$.}
     \label{ConditionPrior}
\end{align}
The prior does not need to depend on $x$; we may choose
$\pi_{\rm prior}\left(\alpha \, \,| \,x \right) = \pi_{\rm prior}(\alpha)$, which may be easier to specify in practice, but we 
allow for general priors $\pi_{\rm prior}\left(\alpha \, \,| \,x \right)$ in the following.

\subsubsection{Posterior distribution of the fixed effects}
 Given the chosen prior $\pi_{\rm prior}$, the posterior distribution of $A$, conditional on $Y=y$ and $X=x$, for given $\theta$, is
\begin{align}
   \pi_{\rm post}(\alpha \,|\, y,x,\theta) 
     = \frac{ f(y \, | \,x,\alpha,\theta) \, \pi_{\rm prior}(\alpha\,|\,x)} {p_{\rm prior}(y \, | \, x,\theta)} \, ,
     \label{DefPost}
\end{align}     
where $p_{\rm prior}(y \, |\, x,\theta) = \int_{\cal A}  f( y \,  | \,x,\theta,  \alpha) \, \pi_{\rm prior}( \alpha\,|\,x) \, {\rm d}  \alpha$ is the  prior predictive probability of outcome $y$.

\subsubsection {Score function}
Let $s \, : \, {\cal Y} \times {\cal X} \times \theta \rightarrow   \mathbb{R}^{d_s}$ be some function, which we will call the ``score function''.
In our numerical illustrations in Section~\ref{sec:Numerical2} we 
 choose the integrated score
\begin{align}
    s(y,x,\theta) &=   \nabla_\theta  \left[  \log \int_{\cal A} f(y \,|\, x, \alpha,\theta) \pi_{\rm prior}(\alpha\,|\,x) \, {\rm d} \alpha \right] 
   \nonumber \\
      &=  \int_{\cal A} \left[ \nabla_\theta \, \log f(y \,|\, x, \alpha,\theta) \right]  \pi_{\rm post}(\alpha \,|\, y,x,\theta)  \, {\rm d} \alpha \, ,
      \label{IntegratedScore}
\end{align}
where $d_s=d_\theta$.
However, for our construction of moment functions 
in the following subsection, which is 
based on a chosen score function, we can actually choose almost any function $s(y,x,\theta)$,
as long as it is differentiable in $\theta$
and not identically zero.
For example, the profile score 
 $s(y,x,\theta) =  \nabla_\theta \left[ \max_{\alpha\in\cal{A}} \log f(y | x, \alpha,\theta) \right] $
would  be an equally natural choice.

Whatever the choice of $s(y,x,\theta)$, we now rewrite it using matrix notation.
Let $n_{\cal Y} = | {\cal Y}|$ be the number of possible outcomes, and label 
the elements of the outcome set as $y_{(k)}$, $k=1,\ldots,n_{\cal Y}$,
so that ${\cal Y}=\left\{ y_{(1)},\ldots,y_{(n_{\cal Y})} \right\}$.
For $y \in {\cal Y}$, let $\delta(y) = (\delta_1(y),\ldots,\delta_{n_{\cal Y}}(y))'$ be the 
$n_{\cal Y}$-vector with   elements  $\delta_k(y) = \mathbbm{1}(y = y_{(k)})$, $k=1,\ldots,n_{\cal Y}$,
where $ \mathbbm{1}(\cdot)$ is the indicator function.
Recall that $s(y,x,\theta) $ is a $d_s$-vector.
Let $S(x,\theta)$  be the corresponding  $d_s  \times n_{\cal Y}$ matrix with columns $s(y_{(k)},x,\theta)$,
$k=1,\ldots,n_{\cal Y}$.  Now we can write $s(y,x,\theta)$ as
\begin{align}
    s(y,x,\theta) = S(x,\theta) \, \delta(y) \, , \qquad \textrm{for all }y\in\cal{Y} \, .
\end{align}

\subsubsection {Posterior predictive probability matrix}  
 Given $x \in {\cal X}$ and $\theta \in \Theta$, after observing some 
$y \in {\cal Y}$, the posterior predictive probability of observing any
``future" $\widetilde y\in {\cal Y}$ is
\begin{align}
    Q( \widetilde y \, | \, y,x,\theta) = \int_{\cal A} f \left(\widetilde y \, \big| \, x,\alpha,\theta \right)  \pi_{\rm post} \left( \alpha \, \big| \,  y,x,\theta \right)  {\rm d} \alpha \, .
    \label{DefQ}
\end{align}
Let $Q(x,\theta)$ be the $n_{\cal Y} \times n_{\cal Y}$ matrix with elements $Q_{k,\ell}(x,\theta)  = Q(y_{(k)} \, | \, y_{(\ell)},x,\theta)$. The following lemma
states some  properties of the matrix $Q(x,\theta) $ that will be useful later.

\begin{lemma}
    \label{lemma:MatrixQ}
     Let $x \in {\cal X}$ and $\theta \in \Theta$.
     Assume that $p_{\rm prior}(y \, |\, x,\theta) > 0$ for all $y \in {\cal Y}$.
Then $Q(x,\theta)$ is
diagonalizable and all its eigenvalues are  real numbers in the interval $[0,1]$.
\end{lemma}

The proof of the lemma is given in the appendix.

\subsection{Bias-corrected score functions}   

We consider the chosen score function $ s(y,x,\theta) $ as our first candidate for a moment function $m(y,x,\theta)$ to achieve \eqref{moment2approx}. However, 
 $\mathbb{E}   [s(Y,X,\theta_0)]$ need neither be zero nor close to zero.
(As discussed above, there are choices of $ s(y,x,\theta) $ such that
$\mathbb{E}   [s(Y,X,\theta_0)]$ is close to zero for large $T$, but even
for those choices $\mathbb{E}   [s(Y,X,\theta_0)]$ need not be close to zero for small $T$.)
  
Therefore, we aim to ``bias-correct'' the score by defining an improved score
$$s^{(1)}(y,x,\theta) := s(y,x,\theta) -  b(y,x,\theta) \, ,$$ for some correction function $b(y,x,\theta)$. The goal is to choose $b(y,x,\theta)$ such that 
the elements of
 $\mathbb{E}   [s^{(1)}(Y,X,\theta_0)] $
are smaller than those of
 $\mathbb{E}   [s(Y,X,\theta_0)]$.
According to the model we have
\begin{align}
     \mathbb{E}\left[ s(Y,X,\theta_0) \, \big| \, X=x, \, A=\alpha \right]
     &= \sum_{y \in {\cal Y}}  s(y,x,\theta_0)  \,  f\left(y \, \big| \, x, \alpha, \theta_0 \right)  \, ,
  \nonumber \\
     \mathbb{E}\left[ s(Y,X,\theta_0) \, \big| \, X=x \right]
     &= \sum_{y \in {\cal Y}}  s(y,x,\theta_0)  \, \int_{\cal A} \,   f\left(y \, \big| \, x, \alpha, \theta_0 \right) \, \pi_0(\alpha \,| \,x) \, {\rm d}\alpha \, .
     \label{InfeasibleBiasCorrection}
\end{align}
We  would achieve
exact bias correction (i.e., $\mathbb{E}   [s^{(1)}(Y,X,\theta_0)] =0$) 
if we could choose $ b(y,x,\theta)$
(or its conditional expectation) equal to
$  \mathbb{E}[ s(Y,X,\theta_0) \, \big| \, X=x, \, A=\alpha  ]$
or equal to $   \mathbb{E}[ s(Y,X,\theta_0) \, \big| \, X=x ]$.
The first option is infeasible because $A$ is unobserved. The second option is infeasible
because $\pi_0(\alpha \,| \,x)$ is unknown.

However, even though $A$ is unobserved, the posterior distribution
$ \pi_{\rm post}(\alpha \,|\, y,x,\theta) $ should contain some useful information
about $A$. Inspired by \eqref{InfeasibleBiasCorrection} we
suggest choosing
$$
b(y,x,\theta) = \sum_{\widetilde y \in {\cal Y}}  s(\widetilde y,x,\theta)  \, \int_{\cal A} \,   f\left(\widetilde y \, \big| \, x, \alpha, \theta \right) \,
      \pi_{\rm post}(\alpha  \, \big| \, y,x,\theta)   \, {\rm d} \alpha \, ,
$$
where we have replaced $\pi_0(\alpha \,| \,x)$ by $  \pi_{\rm post}(\alpha \,|\, y,x,\theta) $.
 This gives the bias-corrected score
\begin{align}
    s^{(1)}(y,x,\theta) &:= s(y,x,\theta) -  \sum_{\widetilde y \in {\cal Y}}  s(\widetilde y,x,\theta)  \, \int_{\cal A} \,   f\left(\widetilde y \, \big| \, x, \alpha, \theta \right) \,
      \pi_{\rm post}(\alpha  \, \big| \, y,x,\theta)   \, {\rm d} \alpha
    \nonumber  \\
      &= s(y,x,\theta) -  \sum_{\widetilde y \in {\cal Y}}  s(\widetilde y,x,\theta) \, Q(\widetilde y \, | \, y,x,\theta) 
     \nonumber \\
      &=   S(x,\theta) \, \left[ \mathbbm{I}_{n_{\cal Y}} -  Q(x,\theta) \right] \, \delta(y) \, ,
      \label{S1Formula}
\end{align}
where $\mathbbm{I}_{n_{\cal Y}}$ is the $n_{\cal Y}\times n_{\cal Y}$ identity matrix.
Now, if the expected posterior distribution
$ \mathbb{E}[   \pi_{\rm post}(\alpha  \, \big| \, Y,X, \theta_0) \, \big| \, X=x ]  $
 is a good approximation to $\pi_0(\alpha \,| \,x)$, then we expect that
 $ \mathbb{E} [s^{(1)}(Y,X,\theta_0)] $ is indeed
 smaller than $\mathbb{E}   [s(Y,X,\theta_0)] $, that is, $s^{(1)}(y,x,\theta)$ should be a better choice than $ s(y,x,\theta)$ as a moment function $m(y,x,\theta)$ satisfying  \eqref{moment2approx}.  Note also that in the very special case where the prior equals the true distribution of $A$ (i.e., $\pi_{\rm prior} = \pi_0$), 
  $\mathbb{E}[   \pi_{\rm post}(\alpha  \, \big| \, Y,X,\theta_0) \, \big| \, X=x ]= \pi_0(\alpha \,| \,x) $
 and hence
  $s^{(1)}(Y,X,\theta_0)$ has exactly zero mean regardless of the choice of initial score function.
 
Of course, generically we still expect that $\mathbb{E}   [s^{(1)}(Y,X,\theta_0)] \neq 0$. It is, therefore, natural to iterate the above procedure,
 that is, to apply the same bias-correction method that we applied to
 $ s(y,x,\theta) $ also to $  s^{(1)}(y,x,\theta)$, which gives $  s^{(2)}(y,x,\theta)$,
and to continue iterating this procedure. 
Since the bias-correction method applied  to
 $s(y,x,\theta) = S(x,\theta) \, \delta(y)$ gives \eqref{S1Formula},
 it is easy to see that
after $q \in \{1,2,\ldots\}$ iterations of the same bias-correction procedure we  obtain
\begin{align}
       s^{(q)}(y,x,\theta) &:=  S(x,\theta) \left[ \mathbbm{I}_{n_{\cal Y}} -  Q(x,\theta) \right]^q \, \delta(y) \, .
    \label{GeneralIteration}
\end{align}
The bias-corrected functions $s^{(q)}(y,x,\theta)$ are the main choices of
moment function $m(y,x,\theta)$ that we consider in this paper.

To our knowledge, the  bias-corrected  scores  $s^{(q)}(y,x,\theta)$  have not previously been discussed in the  literature.
However, as will be explained in Section~\ref{subsec:AlternativeBC},
these bias-corrected scores are closely related to existing bias-correction methods. In particular, 
if in \eqref{S1Formula} we replace the posterior distribution
$ \pi_{\rm post}(\alpha  \, \big| \, y,x,\theta) $ with a 
point-mass distribution at the MLE of $A$
for fixed $\theta$, then the profile-score
adjustments of \cite{dhaene2015profile} are obtained.
In analogy to such existing methods,
we make the following conjecture,
which
is supported by our numerical results in Section~\ref{sec:Numerical2}
below
for the model in Examples~\ref{ex:StaticPanel}
and \ref{ex:StaticPanel2},
and which we presume to hold more generally
under appropriate 
regularity conditions.

\begin{conjecture}
     \label{Rate}
If we choose the original score function $s(y,x,\theta)$ 
to be the integrated or profile score, then  both
$  \mathbb{E}\left[ s^{(q)}(Y,X,\theta_0)  \right]$
and $\theta_* - \theta_0$ are  at most of order $T^{-q-1}$, as
$T\to\infty$ while $q$ is fixed. 
\end{conjecture}

From the perspective of the large-$T$ panel literature, we
have simply provided another way to achieve and iterate large-$T$ bias
correction. What is truly novel, however,
is that in the limit $q \rightarrow \infty$,
our correction can be directly related to the functional differencing method
of \cite{bonhomme2012functional}, which delivers 
exact (finite-$T$) inference results in models to which it is applicable.
Our procedure, therefore, interpolates  between large-$T$ bias correction and exact finite-$T$ inference.

\paragraph{Remark 1:}
If the initial score function $s(y,x,\theta)$  is  unbiased,
that is, if it satisfies the exact moment condition \eqref{moment}
or, equivalently, \eqref{moment2}, 
then the bias correction step \eqref{S1Formula} does not change the score
function at all. 
More generally, if at any point in the iteration procedure \eqref{GeneralIteration}
we obtain an unbiased score function, i.e., one for which 
$\sum_{y \in {\cal Y}} \, s^{(q)}(y , x, \theta) \, \allowbreak f\left(y \, \big| \, x, \alpha, \theta \right) \allowbreak  = 0$ for some $q$, then  we have
$s^{(r)}(y,x,\theta)  = s^{(q)}(y,x,\theta)$ for all further iterations $r \geq q$.
Hence, unbiased score functions correspond to fixed points
in our iteration procedure.

\subsection{Relation to functional differencing} 
\label{subsec:ExactFD}

It turns out that exact functional differencing
 (\citealt{bonhomme2012functional})
 corresponds to choosing 
 $$ 
   s_{\infty}(y,x,\theta) := \lim_{q \rightarrow \infty} 
    s^{(q)}(y,x,\theta) 
 $$
 as moment functions for estimating $\theta_0$,
and the lemma below formalizes this relationship.

Before presenting the lemma, it is useful to rewrite the bias-corrected score function in \eqref{GeneralIteration} in terms of the spectral decomposition of  $Q(x,\theta)$.
Let $\lambda_1(x,\theta) \ge\ldots\ge \lambda_{n_{\cal Y}}(x,\theta) $ be the eigenvalues of 
$Q(x,\theta)$ sorted in descending order,
and let $U(x,\theta)$ be the $n_{\cal Y} \times n_{\cal Y}$ 
matrix whose columns are the corresponding  right-eigenvectors of $Q(x,\theta)$.
Lemma~\ref{lemma:MatrixQ} guarantees that $\lambda_k(x,\theta) \in [0,1]$,
 for all $k=1,\ldots,n_{\cal Y}$, and that $Q(x,\theta)$ 
 can be diagonalized, that is,
\begin{align*}
     Q(x,\theta) &= U(x,\theta)  \;  \diag[  \lambda_k(x,\theta)  ]_{k=1,\ldots,n_{\cal Y}} \; U^{-1}(x,\theta) \, .
\end{align*}
Now let $h:[0,1] \rightarrow \mathbb{R}$ be a stem function and 
let $h(\cdot)$ be the associated primary matrix function (\citealt{horn1994topics}), so that 
 \begin{align}
    h[ Q(x,\theta) ]
     &=  U(x,\theta)  \;   \diag \left\{ h[  \lambda_k(x,\theta)  ]  \right\}_{k=1,\ldots,n_{\cal Y}} \; U^{-1}(x,\theta) \, .
    \label{DefhQ} 
\end{align}
That is, applying $h(\cdot)$ to the matrix $Q(x,\theta)$ simply means applying the stem function separately to each eigenvalue of 
$Q(x,\theta)$ while leaving the eigenvectors unchanged.
Every stem function $h$ defines a moment function
\begin{align}
       s_h(y,x,\theta) &:=  S(x,\theta)  \;  h[ Q(x,\theta) ] \; \delta(y) \, ,
     \label{GeneralScoreFunctions}
\end{align}
hence generalizing \eqref{GeneralIteration}.
In particular, the stem function $h_q(\lambda) :=  (1-\lambda)^q$ gives the moment function $s^{(q)}(y,x,\theta) =     s_{h_q}(y,x,\theta)$
in \eqref{GeneralIteration}.
In the limit $q \rightarrow \infty$
we obtain
$\lim_{q \rightarrow \infty} \, h_q(\lambda) =  \mathbbm{1}\{  \lambda = 0 \} $, 
for $\lambda \in [0,1]$,
and the limiting bias-corrected score function can therefore be written as
\begin{align}
    s_{\infty}(y,x,\theta) &=  S(x,\theta)  \,   h_\infty[ Q(x,\theta) ] \, \delta(y) \, ,
   &
   h_\infty(\lambda) &:=    \mathbbm{1}\{  \lambda = 0 \}  \, .
   \label{SCOREinf}
\end{align}  
Thus, $ s_{\infty}(y,x,\theta)$ is obtained by applying the projection
$h_\infty[ Q(x,\theta) ]$ to the original
score function $ s(y,x,\theta)$.
The projection matrix $h_\infty[ Q(x,\theta) ]$ is obtained
according to \eqref{DefhQ}
 by giving weight only to eigenvectors of $Q(x,\theta)$ accociated with zero eigenvalues.
\begin{lemma}
     \label{ExactFD}
      Let $x \in {\cal X}$. 
     Suppose that
     \eqref{model} and     
     \eqref{ConditionPrior} hold, 
          that $p_{\rm prior}(y \, |\, x,\theta_0) > 0$ for all $y \in {\cal Y}$,
     and that $f\left(y \, \big| \, x, \alpha, \theta_0 \right)$ is continuous
     in $\alpha\in{\cal A}$.
     Then
\begin{itemize}
   \item[(i)] we have
   $$\mathbb{E} \left[  s_{\infty}(Y,X,\theta_0)  \, \big| \, X=x, \, A = \alpha  \right] = 0 \, ,
    \qquad \textrm{for all } \alpha \in {\cal A} \, ;$$
   \item[(ii)]
    the matrix $Q(x,\theta_0)$
    has a zero eigenvalue if and only if
    there exists a non-zero moment function
   ${\mathfrak m}(y , x, \theta_0) \in \mathbb{R}$
   that satisfies
    $$ \mathbb{E} \left[ {\mathfrak m}(Y , X, \theta_0) \, \big| \, X=x, \, A = \alpha  \right] = 0 \, ,
    \qquad \textrm{for all } \alpha \in {\cal A} \, ;
    $$    
   \item[(iii)]  
   for every moment function
   ${\mathfrak m}(y , x, \theta_0) \in \mathbb{R}$ 
    that satisfies the condition in part (ii),  
   there exists a function $s(y , x, \theta_0) \in \mathbb{R}$
   such that
   $${\mathfrak m}(y , x, \theta_0) = s_{\infty}(y,x,\theta_0) \, ,\qquad  \textrm{for all } y \in {\cal Y} \,.$$
   
\end{itemize}
\end{lemma}

The proof of the lemma is given in the appendix.
Note that the true parameter value, $\theta_0$, only takes a special role in Lemma~\ref{ExactFD} because 
the expectation $\mathbb{E} \left(  \cdot \, | \, X=x, \, A = \alpha  \right)$ is evaluated using
$ f(y | x, \theta_0,\alpha)  $, according to \eqref{model}. If we had written these conditional expectations as explicit
sums over $ f(y | x, \theta_0,\alpha)  $, then we could have replaced $\theta_0$ in the lemma by an arbitrary value $\theta \in \Theta$; that is, there is nothing special about the parameter value $\theta_0$ that generates the data.

Part (i) of the lemma states that $s_{\infty}(y,x,\theta) $ is an exactly valid moment function in the sense of \eqref{moment}.
If  $Q(x,\theta)$ does not have any 
zero eigenvalues, then this part of the lemma is a trivial result, 
because then we simply have $ s_{\infty}(y,x,\theta)=0$, 
which is not useful for estimating $\theta_0$. However, if  $Q(x,\theta)$ does
have one or more zero eigenvalues, then, for a generic choice of $s(y,x,\theta)$, we have
$ s_{\infty}(y,x,\theta) \neq 0$, and part (i) of the lemma
becomes non-trivial.

Part (ii) of the lemma states that the existence of a
zero eigenvalue of $Q(x,\theta)$ is indeed a necessary and
sufficient condition for the existence of an
exactly valid moment function in the sense of \eqref{moment}.
As explained in the proof, if $ Q(x,\theta)$ has a zero eigenvalue, then
an exactly valid moment function ${\mathfrak m}(y , x, \theta)$  is simply obtained
by the entries of the corresponding left-eigenvector of $ Q(x,\theta)$.

Finally, part (iii) of the lemma states that any
such exactly valid moment function ${\mathfrak m}(y , x, \theta)$
can be obtained as 
$ \lim_{q\to\infty}s^{(q)}(y,x,\theta) $, i.e., as the limit of our iterative bias correction 
scheme above, for some appropriately chosen
initial score function $s(y , x, \theta)$. Thus, the set
of valid moment functions is identical to 
the set of all possible limits $ s_{\infty}(y,x,\theta) $.

Recall that finding such exactly valid moment functions
is the underlying idea of the functional differencing
method of \cite{bonhomme2012functional}. Thus, 
Lemma~\ref{ExactFD} establishes a very close relationship
between our bias correction method and functional differencing.

\paragraph{Remark 2:} 
If the set ${\cal A}$ is finite
with cardinality $n_{\cal A}=|{\cal A}|$,
then by construction 
${\rm rank}[Q(x,\theta)]\leq n_{\cal A}$. Thus,
whenever $n_{\cal A} < n_{\cal Y}$,
$Q(x,\theta)$ has $n_{\cal Y}- n_{\cal A}$ zero eigenvalues, 
implying that exact moment functions, free of $\alpha$,
are available. Notice, however, that this
assumes not only that $\alpha$ takes on only a finite
number of values, but also that these values
are known (they constitute the known set ${\cal A}$). 
By contrast, the literature on 
discretizing heterogeneity in panel data (e.g., 
\citealt{bonhomme2015grouped},
\citealt{su2016identifying},
\citealt{bonhomme2022discretizing}) 
usually considers the support points
of $A$ to be unknown.
For our purposes,
the fact that ${\rm rank} [Q(x,\theta)] \leq n_{\cal A}$ matters
only in our numerical implementation,  where
the rank of $Q(x,\theta)$ might be
truncated my the discretization of the set ${\cal A}$.

\section{Eigenvalues of $Q(x,\theta)$: numerical example} 
\label{sec:Numerical1}

Lemma~\ref{lemma:MatrixQ} guarantees that all
eigenvalues of the matrix $Q(x,\theta)$ lie in the interval $[0,1]$, 
and Lemma~\ref{ExactFD} shows that exact moment conditions
that are free of the incidental parameter $A$ are only
available if $Q(x,\theta)$ has a zero eigenvalue.
However, even in models where $Q(x,\theta)$ does not have
a zero eigenvalue, we suggest that calculating
the eigenvalues of $Q(x,\theta)$ is generally informative
about whether moment conditions exist that are approximately
free of the incidental parameters. This is because in
typical applications we expect that the distinction between
a zero eigenvalue and a very small eigenvalue of $Q(x,\theta)$ 
should be practically irrelevant, that is, as long as $Q(x,\theta)$
has one or more eigenvalues
that are very close to zero, then very good
approximate moment conditions in the sense of \eqref{moment2approx}
should exist.

It is difficult to make a general statement
about how small an eigenvalue of $Q(x,\theta)$ needs to be to
qualify as sufficiently small. However, in a typical
model with a sufficiently large number $n_{\cal Y}$ of outcomes  
(which for discrete choice panel data usually requires only moderately large $T$)
one will often have multiple eigenvalues of $Q(x,\theta)$ that
are so small (say smaller than $10^{-5}$) that there is little
doubt that they can be considered equal to zero for practical
purposes.

To illustrate this, consider Example~\ref{ex:StaticPanel2}
with normally distributed errors, $F(u)=\Phi(u)$,
even values of $T$, and $T_0=T_1=T/2$, which implies
$n_{\cal Y} = (1+T/2)^2$.
For the prior distribution of $A$ we  choose the standard normal distribution, $\pi_{\rm prior}(\alpha) = \phi(\alpha)$.\footnote{
In fact, for our numerical implementation, we discretize the standard normal prior by choosing 1000 grid points
$\alpha_j = \Phi^{-1}(j/1001)$, $j=1,\ldots,1000$, and we implement a prior
that gives equal probability to each of these grid points. The approximation
bias that results from this discretization is negligible for our purposes,
as long as $n_{\cal Y}$ is much smaller than 1000.
}
We then calculate the eigenvalues of the $n_{\cal Y} \times n_{\cal Y}$ matrix $Q(\theta)$
for  $\theta=1$ (there are no longer covariates $x$ in this example as they
are assigned non-random values). 
For $T=2$ we have $n_{\cal Y}=4$, and the four eigenvalues of $Q(1)$ are
$\lambda_1=1$, $\lambda_2=0.47463$, $\lambda_3=0.10727$, and $\lambda_4 = 0.00016$. For $T=4$ and $T=6$ we have $n_{\cal Y}=9$ and $n_{\cal Y}=16$,
respectively, and the corresponding eigenvalues of $Q(1)$
are plotted in Figures~\ref{EigenvaluePlotExample}
and~\ref{EigenvaluePlotExample2}. Figure~\ref{EigenvaluePlotExample3}
plots only the smallest eigenvalues of $Q(1)$ for $T=2,4,\ldots,20$.

\begin{figure}[p]
\begin{center}
\begin{tikzpicture}[scale=0.72]
\begin{axis}[
    xtick={1,2,3,4,5,6,7,8,9},
    xlabel={$j$},
    ylabel={eigenvalue $\lambda_j$},
    ymin=0
]
\addplot[only marks] table {
1 1.000000000000000000000000000000
2 0.644201524692854103911090391973
3 0.283013202537976863007143247119
4 0.076399117353885997064153145562
5 0.010121540641057413782706802501
6 0.000154754552155926706424821198
7 0.000033960047375582489859428991
8 0.000000078736447438692020901689
9 0.000000000756251178799830909244
};
\end{axis}
\end{tikzpicture}
\qquad \qquad
\begin{tikzpicture}[scale=0.72]
\begin{axis}[
    ymode=log,
    xtick={1,2,3,4,5,6,7,8,9},
    xlabel={$j$},
    ylabel={eigenvalue $\lambda_j$}
]
\addplot[only marks] table {
1 1.000000000000000000000000000000
2 0.644201524692854103911090391973
3 0.283013202537976863007143247119
4 0.076399117353885997064153145562
5 0.010121540641057413782706802501
6 0.000154754552155926706424821198
7 0.000033960047375582489859428991
8 0.000000078736447438692020901689
9 0.000000000756251178799830909244
};
\end{axis}
\end{tikzpicture}
\vspace{-0.2cm}
\caption{\label{EigenvaluePlotExample}
The eigenvalues $\lambda_j(\theta)$ of the matrix $Q(\theta)$ 
in Example~\ref{ex:StaticPanel2} are plotted 
for the case $\theta=1$,
$T=4$, $T_0=T_1=2$, and where both the error distribution and the prior
distribution of $A$ are standard normal.
The left and right plots show the same eigenvalues, just with 
a different scaling of the y-axis.
}
\end{center}
\end{figure}
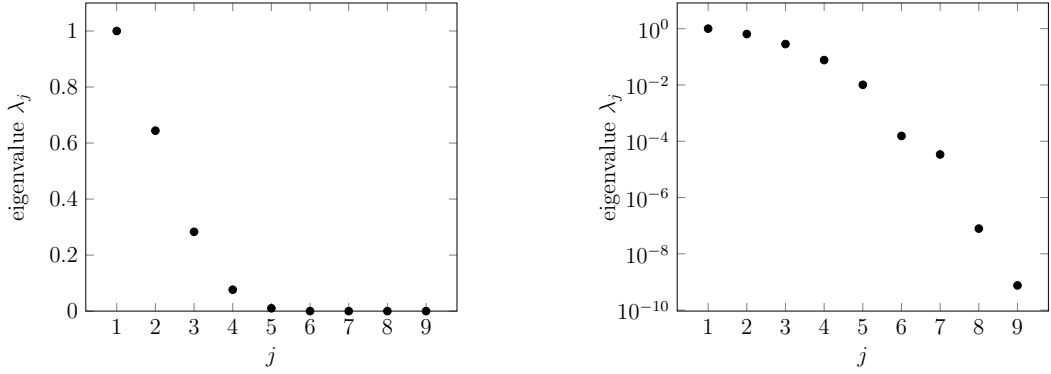

\begin{figure}[p]
\begin{center}
\begin{tikzpicture}[scale=0.72]
\begin{axis}[
    xlabel={$j$},
    ylabel={eigenvalue $\lambda_j$},
    ymin=0
]
\addplot[only marks] table {
1 1.000000000000000000000000000000
2 0.731096535483503931247780179272
3 0.409921104247830183138863350110
4 0.170769502676008130955222655012
5 0.050392351409742479203816615997
6 0.009623824615960652010576908032
7 0.000996153481709224824532565980
8 0.000050321684786527456361949998
9 0.000009620879592293233117570132
10 0.000000136669846492106888788686
11 0.000000027714527436872997416895
12 0.000000001254275321212861831022
13 0.000000000158052389368709712392
14 0.000000000001056341159198913450
15 0.000000000000008654370222156573
16 0.000000000000000000196653048839
};
\end{axis}
\end{tikzpicture}
\qquad \qquad
\begin{tikzpicture}[scale=0.72]
\begin{axis}[
    ymode=log,
    xlabel={$j$},
    ylabel={eigenvalue $\lambda_j$}
]
\addplot[only marks] table {
1 1.000000000000000000000000000000
2 0.731096535483503931247780179272
3 0.409921104247830183138863350110
4 0.170769502676008130955222655012
5 0.050392351409742479203816615997
6 0.009623824615960652010576908032
7 0.000996153481709224824532565980
8 0.000050321684786527456361949998
9 0.000009620879592293233117570132
10 0.000000136669846492106888788686
11 0.000000027714527436872997416895
12 0.000000001254275321212861831022
13 0.000000000158052389368709712392
14 0.000000000001056341159198913450
15 0.000000000000008654370222156573
16 0.000000000000000000196653048839
};
\end{axis}
\end{tikzpicture}
\vspace{-0.2cm}
\caption{\label{EigenvaluePlotExample2}
Same eigenvalue plot as in Figure~\ref{EigenvaluePlotExample}, but for
 $T=6$ and $T_0=T_1=3$.
}
\end{center}
\end{figure}
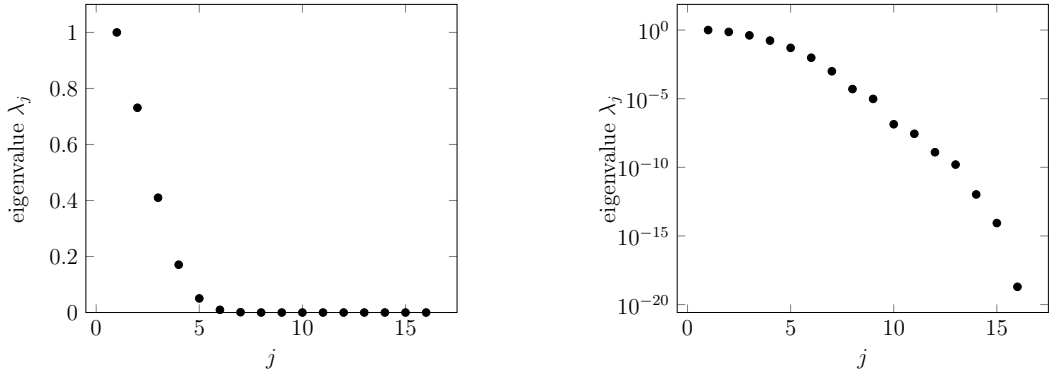

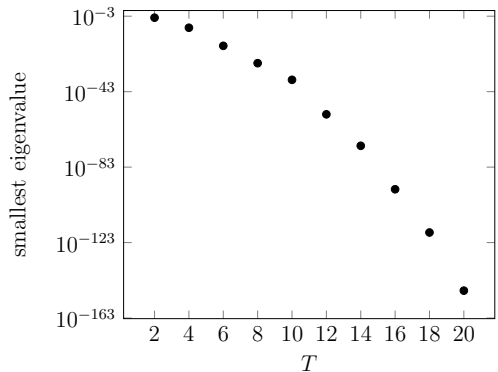
\begin{figure}[p]
\begin{center}
\begin{tikzpicture}[scale=0.72]
\begin{axis}[
    ymode=log,
    xlabel={$T$},
    ylabel={smallest eigenvalue},
    ymax=1,
    y label style={at={(-0.1,0.5)}},
    xtick={2,4,6,8,10,12,14,16,18,20},
]
\addplot[only marks] table {
2 0.00016208962688619187757609173834544867487675901462679219249227836731737711459070490218225910348408958762385461562826280840781251440724258295910989836440396844304824591899601844668688750946052015772937
4 0.00000000075625117879983090924384138287028416800414705667406799135688462835877015852273946463660154060182408080700412515947591894666040203842943212290200159801201144010738314433670516864030396967442678
6 0.00000000000000000019665304883857398392959902340535295197322483323113362230231833559022232397057934651079479104056451878496856609101892012316037844474602063847308853132148933444335192228025714360512271
8 0.00000000000000000000000000012446362894628706415113250017550411946330822286631245807652905256311042510908645521625366118972405082036238340671106887372851596926767068083762666687212519810753992036693565
10 0.00000000000000000000000000000000000018904666806866142200307323981191767948645274607617728170156045848235293088994449304679570984631404149037000005656699782439344139253540574857785197889913378593306673
12 0.00000000000000000000000000000000000000000000000000000009547273828090526347347377393180742384069283060322755828279245592707727251887546079394096921126433650810080600180020185166461015703279055561081342
14 0.00000000000000000000000000000000000000000000000000000000000000000000000203461147003047724089573013584213544419195716006725984899125914272604732488818200946321493744413836693866260565512516290492727233
16 0.00000000000000000000000000000000000000000000000000000000000000000000000000000000000000000000001806764038426894757156288915631551060349805777378444274679354601225081686975792385060346686462141157461667
18 0.00000000000000000000000000000000000000000000000000000000000000000000000000000000000000000000000000000000000000000000021890103068014099408475901619489323867474768850240678122818458185400338702428437386
20 0.00000000000000000000000000000000000000000000000000000000000000000000000000000000000000000000000000000000000000000000000000000000000000000000000000003384614864471714923893879891466374933163629840691317
};
\end{axis}
\end{tikzpicture}
\vspace{-0.2cm}
\caption{\label{EigenvaluePlotExample3}
For the same setting as in Figure~\ref{EigenvaluePlotExample}, but for
different values of $T$ (with $T_0=T_1=T/2$), we plot only the smallest
eigenvalue of $Q(\theta)$ for $\theta=1$. Notice that the smallest
eigenvalue is never zero, that is, $Q(1)$ has full rank
for all values of $T$ considered.
}
\end{center}
\end{figure}

From these figures, we see that for $T \geq 4$ the smallest eigenvalue
of $Q(1)$ is less than $10^{-9}$,
which we argue can be considered equal to zero for practical purposes. In Figures~\ref{EigenvaluePlotExample}
and~\ref{EigenvaluePlotExample2} we see that the largest
eigenvalue, $\lambda_1$, is equal to one (because $Q(1)$
is a stochastic matrix), but then the eigenvalues $\lambda_j$
decay exponentially fast as $j$ increases.\footnote{
The eigenvalues of $Q(1)$ presented in this section were obtained using Mathematica with a numerical precision of 1000 digits.}

If we were to replace
the standard normal distribution of the errors 
by the standardized logistic cdf
    $F(u)=   (1+e^{-\pi u / \sqrt{3}})^{-1}$
(normalized to have variance one),     
then the left-hand side (non-logarithmic) plots 
in Figures~\ref{EigenvaluePlotExample}
and~\ref{EigenvaluePlotExample2}
would look almost identical,
but there would be $(T/2)^2$ eigenvalues exactly equal to zero.
These zero eigenvalues for the logit model
are  due to the existence of a sufficient
statistic for $A$ and, correspondingly, the existence of
exact moment functions (associated with the left null-space
of $Q(\theta)$), as discussed in Section~\ref{subsect:ExactMoments}
above. Given that the change from 
standardized logistic errors to 
standard normal errors is a relatively
minor modification of the model, it is not
surprising that we see many eigenvalues close to zero 
in Figures~\ref{EigenvaluePlotExample}
and~\ref{EigenvaluePlotExample2}.

Figure~\ref{EigenvaluePlotExample3} shows that the smallest
eigenvalue of $Q(1)$ in this example also decays exponentially fast
as $T$ increases.\footnote{Presumably this finding
and the fast shrinkage of the identified
sets of common parameters (\citealt{HonoreTamer06}) and
average effects (\citealt{chernozhukov2013average})
are manifestations of the same phenomenon.}
However, for none of the values
of $T$ that we considered here, did we find an exact 
zero eigenvalue for the static binary choice probit model.
We conjecture that this is true for all $T\ge 2$, but we have
no proof.\footnote{
One needs to be careful with such conclusions for all $T$.
For example, we also experimented with another error
distribution.
If,  in Example~\ref{ex:StaticPanel2}  with $\theta=1$
and $T_0=T_1=T/2$,
 one chooses the error distribution
$F(u)$ to be the Laplace distribution with mean
zero and scale one, then numerically we found that
for any choice of prior the matrix
$Q(1)$ does not have
a zero eigenvalue for $T=2$ and $T=4$, but it does for $T=6$. So it is not impossible
that something similar could happen for the probit model
for sufficiently large $T$, although we do not expect it.
}

 This example illustrates that eigenvalues of $Q(x,\theta)$
very close to zero but not exactly zero may exist in interesting models.
When aiming to estimate the parameter
$\theta_0$ in a particular model of the form \eqref{model2},
our first recommendation is to calculate the
eigenvalues of $Q(x,\theta)$ for some representative values
of $\theta$ and $x$ to see if some of them are zero
or close to zero. If some are equal to zero, then
exact functional differencing
 (\citealt{bonhomme2012functional}) is applicable.
 If some are very close to zero, then approximate
 moment functions (as in \eqref{moment2approx}) are available.

The eigenvalues of $Q(x,\theta)$ are useful to examine whether exact or approximate moment functions for $\theta$
are available in a given model. However, as explained in Section~\ref{subsect:ExactMoments}, the corresponding 
moment functions also have to depend on $\theta$ to be useful for parameter estimation. For example, the matrix 
$Q(1)$ in Example~\ref{ex:StaticPanel} has exactly the same non-zero eigenvalues as
the matrix $Q(1)$ in Example~\ref{ex:StaticPanel2} that we just discussed,
but in addition, it has a zero eigenvalue with multiplicity equal to
 $2^T - (T_0+1)(T_1+1)$, corresponding to the uninformative moment functions in equation \eqref{uninformativeMoments}.  
 As a diagnostic tool, it can also be useful to  calculate the matrix $Q(x)$ for the model $ f(y | x, \theta_i,\alpha_i) $, which has no common parameters and where both $\theta_i$ and $\alpha_i$ are individual-specific fixed effects (this requires choosing a prior
for $\theta$ as well, which may have finite support to keep the computation simple). Every zero eigenvalue of that matrix $Q(x)$ then corresponds
to a moment function, for that value of $x$, that does not depend on $\theta$ (within the range of the chosen prior for $\theta$).
The  existence of uninformative moment functions \eqref{uninformativeMoments}
in Example~\ref{ex:StaticPanel}, for example,
can be detected in this way.

\section{Estimation}
\label{sec:Estimation}

Suppose we have chosen a prior distribution $\pi_{\mathrm{prior}}$, an initial score function $s(y,x,\theta)$, and an order of bias correction, $q$. This gives
the bias-corrected score function $s^{(q)}(y,x,\theta)$ as the moment function $m(y,x,\theta)$ for which the approximate moment condition 
\eqref{moment2approx} is assumed to hold.
For simplicity, suppose that $d_s=d_{\theta}$, so that 
the number of moment conditions equals the number of common
parameters we want to estimate.
We can then define 
a pseudo-true value $\theta_* \in \Theta$ as the solution 
of
\begin{align}
    \mathbb{E}\left[ m(Y_i,X_i,\theta_*)  \right] = 0 \, .
     \label{moment2pseudo}
\end{align}
The corresponding method of moments estimator $\widehat \theta$ satisfies
$
    \frac 1 n \sum_{i=1}^n \,  m(Y_i,X_i,\widehat \theta)   = 0  .
$
Under appropriate regularity conditions,
including existence and uniqueness of  $\theta_*$, we then have, as $n \rightarrow \infty$,
$$
     \sqrt{n}( \widehat \theta - \theta_* ) \;  \overset{d}{\to}  \; {\cal N}(0,V_*) \, ,
$$
with asymptotic variance given by
\begin{align}
     V_* &= [G'_*]^{-1} {\rm Var}\left[ m(Y_i,X_i,\theta_*) \right] G^{-1}_*  \, ,
     &
     G_* &=  \mathbb{E}\left[ \nabla_{\theta} \, m'(Y_i,X_i,\theta_*) \right] \, .
\end{align}
In Section~\ref{sec:Numerical2} we will report
the bias $\theta_* - \theta_0$ and the asymptotic variance
$V_*$ for different choices of moment functions in
Example~\ref{ex:StaticPanel2}.
Note that reporting the bias $\theta_* - \theta_0$ of the 
parameter estimates is   more informative 
than reporting the bias $\mathbb{E}\left[ m(Y_i,X_i,\theta_0)  \right]$
of the moment condition, in particular since the moment condition
can be rescaled by an arbitrary factor.

How should $q$ be chosen? 
If $Q(x,\theta)$ is singular, then
a natural choice is $q=\infty$,
as described in Section~\ref{subsec:ExactFD}, 
because this delivers an exactly unbiased moment condition.
If $Q(x,\theta)$ is nonsingular but some of its eigenvalues are small,
then, recalling our discussion in Section~\ref{sec:Numerical1},
our general recommendation is
to choose relatively large values of $q$.
The larger the chosen $q$, the more we rely on the smallest eigenvalues of $Q(x,\theta)$, because contributions to $s^{(q)}(y,x,\theta)$ from
larger eigenvalues of $Q(x,\theta)$ are downweighted more heavily as $q$ increases.
If none of the eigenvalues $Q(x,\theta)$ is close to zero, then
there are no moment conditions that hold approximately in the sense of 
\eqref{moment2approx}. Yet, even then, setting $q>0$ is likely to improve on
$q=0$, even though the remaining bias will still be non-negligible in general.
Whatever the eigenvalues of $Q(x,\theta)$ are (and, indeed, whether $Q(x,\theta)$ is singular or not), $q$ is a tuning parameter and a principled
way to choose $q$ would be to optimize some criterion,
for example, the (estimated) mean squared error of $\widehat \theta$.
We leave this for further study. In our numerical illustrations in Section~\ref{sec:Numerical2} we just consider finite values of $q$ up to $q=1000$, and $q=\infty$.

\bigskip

\section{Asymptotic and finite-sample properties}
\label{sec:Numerical2}

In this section, we report on asymptotic and finite-sample properties of $\widehat\theta$ for different choices of moment functions in the model of Example~\ref{ex:StaticPanel2} with standard normal errors (i.e., the panel probit model with a single, binary regressor and fixed effects)
and a variation thereof, the model of Example~\ref{ex:StaticPanel} with a continuous regressor. 
Throughout, we set $\theta_0=1$, we use a standard normal prior (i.e., $\pi_{\mathrm{prior}}(\alpha)=\phi(\alpha)$, as in Figures 1 and 2), we choose the integrated score \eqref{IntegratedScore} as the initial score function, and we vary $q$, the number of iterations
of the bias correction procedure.

We first present results on asymptotic and finite-sample biases and variances of $\widehat\theta$ for three cases where $T$ is relatively small:
Example~\ref{ex:StaticPanel2} with $T_{0}=T_{1}=T/2$ and $T\in\{4,6\}$ (Case 1); 
Example~\ref{ex:StaticPanel2} with $T_{0}=1, T_{1}=T-1$ and $T\in\{4,10\}$ (Case 2);
Example~\ref{ex:StaticPanel} with a continuous regressor 
$X_{it}\sim\mathcal{N}(0.5,0.25)$
and $T\in\{4,6\}$ (Case 3).
In all three cases we set the true distribution of $A$ equal to $\mathcal{N}(1,1)$ (i.e., $\pi_{0}(\alpha)=\phi(\alpha-1)$); note that this implies that $\pi_{\mathrm{prior}}$ is rather different from $\pi_{0}$.

Then, in the setup of Example~\ref{ex:StaticPanel2} with standard normal errors, we numerically explore Conjecture 1 by examining the asymptotic bias,
$\theta_{\ast}-\theta_{0}$, for $T$ up to 512, $q$ up to 3, and various choices of $\pi_{0}$ as detailed below.

\subsection{Case 1: Binary regressor, $T_0=T_1=T/2$}

Table~\ref{BiasAvarExample} reports $\theta_{\ast}-\theta_{0}$ and $V_{\ast}$
for the case where $T_{0}=T_{1}=T/2$ and $T\in\{4,6\}$.
The uncorrected estimate of $\theta_{0}$ $(q=0)$ has a large positive bias,
$0.5050$ when $T=4$ and $0.4056$ when $T=6$. Bias correction ($q>0$) reduces
the bias considerably, though non-monotonically in $q$. The least bias is
attained at $q=\infty$, where the bias is very small:
$-0.52\times10^{-4}$ when $T=4$ and ${-}0.25\times10^{-10}$ when $T=6$.

Our calculation of $V_{\ast}$ shows that there is, overall, a bias-variance
trade-off in this example. When $T=4$, $V_{\ast}$ slightly decreases as we move
from $q=0$ to $q=1$, but for $q\geq1$ we see that $V_{\ast}$ increases in
$q$; when $T=6$, $V_{\ast}$ increases in $q$ throughout. Strikingly, $V_{\ast}$
at $q=\infty$ is much larger than, say, at $q=1000$.\footnote{
The limit $\lim_{q \rightarrow \infty} \theta_*(q)$
can be obtained by solving
$\mathbb{E}[S(\theta_*)U_{n_{\cal Y}}(\theta_*)
[U^{-1}(\theta_*)]_{n_{\cal Y}}\delta(Y)] = 0$ for $\theta_*$, 
where $U_{n_{\cal Y}}(\theta)$
is the submatrix of $U(\theta)$ whose columns are
the right-eigenvectors of $Q(\theta)$ corresponding to
$\lambda_{n_{\cal Y}}(\theta)$, 
the smallest eigenvalue of $Q(\theta)$, and
$[U^{-1}(\theta)]_{n_{\cal Y}}$ is the submatrix of 
$[U^{-1}(\theta)]$ whose rows are the corresponding left-eigenvectors.
}

Table~\ref{BiasAvarExample} also reports, for a cross-sectional sample size of $n=1000$, the
approximate ${\rm RMSE}$ of $\widehat{\theta}$ and the approximate coverage rate of
the 95\% confidence interval with bounds $\widehat{\theta}\pm(n\widehat
{V}_{\ast})^{1/2}\Phi^{-1}(0.975)$, where $\widehat{V}_{\ast}$ is the
empirical analog to $V_{\ast}$. The approximate ${\rm RMSE}$ is calculated as
${\rm RMSE}=(V_{\ast}/n+(\theta_{\ast}-\theta_{0})^{2})^{1/2}$ and the approximate
coverage rate as ${\rm CI}_{0.95}=\Pr[|Z|\leq\Phi^{-1}(0.975)]$ where $Z\sim
\mathcal{N}((\theta_{\ast}-\theta_{0})/(V_{\ast}/n)^{1/2},1)$. Of course,
${\rm RMSE}$ and ${\rm CI}_{0.95}$ follow mechanically from $\theta_{\ast},V_{\ast},n$ and
heavily depend on the chosen $n$. For our choice of $n=1000$, when $T=4$, the
${\rm RMSE}$ is minimized at $q=2$, but this is a rather fortuitous consequence of the
bias having a local minimum (in magnitude) at $q=2$. In practice, when $T$ is
very small we would not recommend choosing $q$ less than $10$, say, because
otherwise, the remaining bias is often non-negligible. From $q=10$ or $20$
onward, the bias and confidence interval coverage rates are reasonably good.
On the other hand, we would also not recommend choosing $q$ to be very large
(including $q=\infty$), one reason being asymptotic variance inflation.

We also conducted a real Monte Carlo simulation, under the exact same setup as
described (and with $n=1000$ in particular). Table~\ref{SimulationsExample}
gives the results, based on 1000 Monte Carlo replications.
The column \textquotedblleft
bias\textquotedblright\ is $\mathbb{E}(\widehat{\theta}-\theta_{0})$ (estimated by
Monte Carlo), and the second column is $n\times \rm{var}(\widehat{\theta})$ (with $\rm{var}(\widehat{\theta})$ estimated by Monte Carlo), to be compared with $V_{\ast}$ in Table~\ref{BiasAvarExample}.
The columns ${\rm RMSE}$ and ${\rm CI}_{0.95}$ are the finite-$n$ ${\rm RMSE}$ and coverage rate.
All the results are close to those in Table~\ref{BiasAvarExample}, confirming that large-$n$
asymptotics provide a good approximation to the finite-$n$ distribution of
$\widehat{\theta}$. Note that Table~\ref{SimulationsExample}
does not report simulation results for $q=\infty$. This is because in some Monte Carlo runs, in particular for $T=6$, it turned out to be too difficult to numerically distinguish between
the smallest and the second smallest eigenvalue of $Q(x,\theta)$ and, therefore,
to reliably select the eigenvector associated with the smallest eigenvalue. This
is another reason not to recommend choosing $q=\infty$.

\subsection{Case 2: Binary regressor, $T_0=1, T_1=T-1$}
 
There is nothing special about the case $T_0=T_1=T/2$, which we just discussed.
Any other $T_0\ge 1$ and $T_1=T-T_0\ge 1$ lead to qualitatively similar results.
We illustrate this for the case $T_0=1$ and $T_1=T-1$. 
Table~\ref{BiasAvarExample2} is similar to Table~\ref{BiasAvarExample} and reports results for $(T_0,T_1)=(1,T-1)$ with $T\in\{4,10\}$. With $T_0=1$ fixed, we find that the bias for $q=0$ is nearly constant in $T$ (and large), while the bias of the bias-corrected estimates $(q>0)$ decreases in $T$. Again, the bias is not monotonic in $q$ (it changes sign) and it becomes very small as $q$ becomes sufficiently large, albeit more slowly than in the case $T_0=T_1=T/2$. Table~\ref{SimulationsExample2} presents the corresponding simulations, showing that the finite-sample results are, again, very close to asymptotic results reported in Table~\ref{BiasAvarExample2}.

\subsection{Case 3: Continuous regressor}

Here we illustrate approximate functional differencing in a panel probit model with a single continuous regressor and fixed effects.
Apart from the continuity of the regressor, the setup is identical to that in Cases 1 and 2 above.
Specifically, we consider Example~\ref{ex:StaticPanel} with $U_{it}\sim\mathcal{N}(0,1)$, $\theta_0=1$, 
$\pi_{\mathrm{prior}}(\alpha)=\phi(\alpha)$, $\pi_{0}(\alpha)=\phi(\alpha-1)$, and the integrated score as initial score.
We set $X_{it}\sim\mathcal{N}(0.5,0.25)$, so that $X_{it}$ has the same mean and variance (across $t$) as the binary regressor in Cases 1 and 2. For $T\in\{4,6\}$ and $n=1000$, we generated a single data set ${X_{it}}$ ($t=1,\ldots,T$; $i=1,\ldots,n$) to form $\mathcal{X}=\{X_1,\ldots,X_n\}$, so the results are to be understood with reference to this
 $\mathcal{X}$.  Table~\ref{ContinuousRegressor} presents the asymptotic biases and variances
 for $q$ up to $1000$.   (We do not consider $q=\infty$ here because, even
 though $\theta_*$ and $\widehat{\theta}$ remain well-defined in the limit $q\to\infty$,
 the limiting values are generically determined by a single $X_i\in\mathcal{X}$, that is, estimation would be based on a single observation~$i$, for which
 $Q(X_i,\theta)$ has the smallest eigenvalue within the
 sample.)
 The results are similar to those in Table~\ref{BiasAvarExample}, albeit the asymptotic biases and variances are  
 somewhat larger (for all $q$). Table~\ref{ContinuousRegressorSimulations} presents the corresponding simulation results, which are in line with those in Table~\ref{ContinuousRegressor}.

\subsection{Numerical calculations related to Conjecture 1}

Table~\ref{Rates} reports bias calculations for larger values of $T$ that support our conjecture about the rate of the bias as $T$ grows.
The model is as in Example~\ref{ex:StaticPanel2} with standard normal errors,
$\pi_{\mathrm{prior}}(\alpha)=\phi(\alpha)$, $T_0=T_1=T/2$, and $\theta_0=1$.
We calculated the bias, $b_T:=\theta_*-\theta_0$, with $\theta_*$ based on the integrated score,
for $q$ up to $3$, $T\in\{64,128,256,512\}$, and for five different choices of $\pi_0$: three degenerate distributions, $\delta_z$, with mass $1$ at $z\in\{0,1,2\}$ (i.e., $A=z$ is constant); and two uniform distributions, $U[0.5,1.5]$ and $U[0,2]$.
(So in all these cases $\pi_{\mathrm{prior}}$ is very different from $\pi_0$.)
Table~\ref{Rates} gives the bias $b_T$ for the chosen values of $T$,
and also the successive bias ratios $b_{T/2}/b_T$ (the three rightmost columns).
 If Conjecture~\ref{Rate} is correct, we should see these ratios converge to $2^{q+1}$ as $T\to\infty$.
For comparability, we also report the bias for $q=0$, where the known rate is confirmed:
$b_{T/2}/b_T$ converges to $2$ for every $\pi_0$, although when $\pi_0=\delta_2$
the convergence to $2$ is not yet quite visible. Presumably, this is because then $\pi_0$ and $\pi_{\mathrm{prior}}$ are quite different, requiring a larger $T$ for the ratio $b_{T/2}/b_T$ to become stable. For $q=1$, $b_{T/2}/b_T$ is seen to converge to $4$ (as conjectured) when $\pi_0\in\{\delta_0,\delta_1,U[0.5,1.5]\}$, while for $\pi_0\in\{\delta_2,U[0,2]\}$ the convergence is less visible. Overall, the picture is a little more blurred for $q=1$ compared to $q=0$. For $q=2$, where $b_{T/2}/b_T$ should converge to $8$, we tend to see this convergence for $\pi_0\in\{\delta_0,\delta_1\}$, although the picture is more blurred; but also here the order of magnitude of $b_{T/2}/b_T$ is in line with convergence to $8$. Finally, for $q=3$, the picture is even more blurred: we tend to see convergence of $b_{T/2}/b_T$ to $16$ only for $\pi_0=\delta_0$, but even here the order of magnitude of $b_{T/2}/b_T$ is not incompatible with convergence to $16$ (apart from the case $\pi_0=U[0,2]$, which clearly needs larger values of $T$ for $b_{T/2}/b_T$ to stabilize). Certainly, these numerical calculations are by no means proof of the conjectured rates, but looking at the last column of Table~\ref{Rates}, the ratios $b_{T/2}/b_T$ are broadly in line with the conjecture. Note, furthermore, that for any $q\ge 1$ the remaining bias in Table~\ref{Rates} is extremely tiny in most cases, unlike the bias of the maximum integrated likelihood estimator (reported  as $q=0$ in the table).

\begin{table}[p]
\begin{center}
\begin{tabular}
[c]{rrrrrrrrrr}\hline\hline
$q$ & $\theta_{\ast}-\theta_{0}$ & \multicolumn{1}{c}{$V_{\ast}$} & $\mathrm{RMSE}$ & $\mathrm{CI}_{0.95}$ &  & $\theta_{\ast}-\theta_{0}$ & \multicolumn{1}{c}{$V_{\ast}$} & $\mathrm{RMSE}$ & $\mathrm{CI}_{0.95}$\\\hline
& \multicolumn{4}{c}{$T_{0}=T_{1}=2,T=4$} &  & \multicolumn{4}{c}{$T_{0}=T_{1}=3,T=6$}\\\cline{2-5}\cline{7-10}
$0$ & $0.5050$ & $3.3313$ & $0.5083$ & $0.0000$ & \multicolumn{1}{r}{} & $0.4056$ & ${2.3063}$ & $0.4084$ & $0.0000$\\
$1$ & $0.1525$ & ${3.3116}$ & $0.1630$ & $0.2452$ & \multicolumn{1}{r}{} & $0.0787$ & $2.3477$ & $0.0924$ & $0.6315$\\
$2$ & $-0.0039$ & $3.4940$ & ${0.0592}$ & $0.9495$ & \multicolumn{1}{r}{} & $-0.0172$ & $2.5114$ & $0.0530$ & $0.9364$\\
$3$ & $-0.0513$ & $3.6583$ & $0.0793$ & $0.8644$ & \multicolumn{1}{r}{} & $-0.0321$ & $2.6289$ & $0.0605$ & $0.9041$\\
$4$ & $-0.0577$ & $3.8016$ & $0.0844$ & $0.8453$ & \multicolumn{1}{r}{} & $-0.0281$ & $2.7157$ & $0.0592$ & $0.9160$\\
$5$ & $-0.0516$ & $3.9307$ & $0.0812$ & $0.8694$ & \multicolumn{1}{r}{} & $-0.0221$ & $2.7789$ & $0.0572$ & $0.9296$\\
$6$ & $-0.0433$ & $4.0438$ & $0.0769$ & $0.8955$ & \multicolumn{1}{r}{} & $-0.0175$ & $2.8231$ & $0.0559$ & $0.9375$\\
$7$ & $-0.0358$ & $4.1383$ & $0.0736$ & $0.9139$ & \multicolumn{1}{r}{} & $-0.0144$ & $2.8532$ & $0.0553$ & $0.9416$\\
$8$ & $-0.0297$ & $4.2145$ & $0.0714$ & $0.9256$ & \multicolumn{1}{r}{} & $-0.0124$ & $2.8735$ & $0.0550$ & $0.9438$\\
$9$ & $-0.0252$ & $4.2745$ & $0.0701$ & $0.9328$ & \multicolumn{1}{r}{} & $-0.0111$ & $2.8873$ & $0.0549$ & $0.9451$\\
$10$ & $-0.0218$ & $4.3210$ & $0.0692$ & $0.9373$ & \multicolumn{1}{r}{} & $-0.0102$ & $2.8969$ & $0.0548$ & $0.9459$\\
$20$ & $-0.0124$ & $4.4722$ & $0.0680$ & $0.9460$ & \multicolumn{1}{r}{} & $-0.0070$ & $2.9262$ & $0.0545$ & $0.9481$\\
$40$ & $-0.0104$ & $4.5149$ & $0.0680$ & $0.9473$ & \multicolumn{1}{r}{} & $-0.0042$ & $2.9365$ & $0.0544$ & $0.9493$\\
$60$ & $-0.0091$ & $4.5311$ & $0.0679$ & $0.9479$ & \multicolumn{1}{r}{} & $-0.0031$ & $2.9392$ & $0.0543$ & $0.9496$\\
$80$ & $-0.0080$ & $4.5415$ & $0.0679$ & $0.9484$ & \multicolumn{1}{r}{} & $-0.0026$ & $2.9404$ & $0.0543$ & $0.9497$\\
$100$ & $-0.0071$ & $4.5495$ & $0.0678$ & $0.9487$ & \multicolumn{1}{r}{} & $-0.0024$ & $2.9411$ & ${0.0543}$ & $0.9498$\\
$200$ & $-0.0046$ & $4.5708$ & $0.0678$ & $0.9495$ & \multicolumn{1}{r}{} & $-0.0020$ & $2.9435$ & $0.0543$ & $0.9498$\\
$400$ & $-0.0033$ & $4.5793$ & $0.0678$ & $0.9497$ & \multicolumn{1}{r}{} & $-0.0017$ & $2.9464$ & $0.0543$ & $0.9499$\\
$600$ & $-0.0031$ & $4.5819$ & $0.0678$ & $0.9498$ & \multicolumn{1}{r}{} & $-0.0015$ & $2.9487$ & $0.0543$ & $0.9499$\\
$800$ & $-0.0030$ & $4.5849$ & $0.0678$ & $0.9498$ & \multicolumn{1}{r}{} & $-0.0014$ & $2.9508$ & $0.0543$ & $0.9499$\\
$1000$ & $-0.0030$ & $4.5881$ & $0.0678$ & $0.9498$ & \multicolumn{1}{r}{} & $-0.0012$ & $2.9528$ & $0.0544$ & $0.9499$\\
$\infty$ & ${-0.0}_{4}{52}$ & $19.2259$ & $0.1387$ & ${0.9500}$ & \multicolumn{1}{r}{} & ${-0.0}_{10}{25}$ &
$13.9013$ & $0.1179$ & ${0.9500}$\\\hline\hline
\end{tabular}
\vspace{-0.2cm}
\caption{\label{BiasAvarExample}
Asymptotic biases and variances, and approximate RMSEs and coverage rates
($n=1000$). The model is as in Example~\ref{ex:StaticPanel2}
with standard normal errors, $\pi_{\mathrm{prior}}(\alpha)=\phi
(\alpha)$, $\pi_{0}(\alpha)=\phi(\alpha-1)$, and $\theta_{0}=1$. The pseudo-true
value $\theta_{\ast}$ is based on the integrated score.
Notation: $0_{r}=\underset{r\text{ zeros}}{\underbrace{00\ldots0}}$, e.g., $-0.0_{4}%
52=-0.000052$. }
\end{center}
\end{table}

\begin{table}[p]
\begin{center}
\begin{tabular}
[c]{rccccccccc}\hline\hline
$q$ & bias & $n\times$var & RMSE & ${\rm CI}_{0.95}$ &  & bias & $n\times$var &
${\rm RMSE}$ & ${\rm CI}_{0.95}$\\\hline
& \multicolumn{4}{c}{$T_{0}=T_{1}=2,T=4$} &  & \multicolumn{4}{c}{$T_{0}%
=T_{1}=3,T=6$}\\\cline{2-5}\cline{7-10}%
$0$ & \multicolumn{1}{r}{$0.5067$} & \multicolumn{1}{r}{$3.5931$} &
\multicolumn{1}{r}{$0.5102$} & \multicolumn{1}{r}{$0.0000$} &
\multicolumn{1}{r}{} & \multicolumn{1}{r}{$0.4076$} &
\multicolumn{1}{r}{${2.3947}$} & \multicolumn{1}{r}{$0.4105$} &
\multicolumn{1}{r}{$0.000$}\\
$1$ & \multicolumn{1}{r}{$0.1543$} & \multicolumn{1}{r}{${3.5243}$} &
\multicolumn{1}{r}{$0.1653$} & \multicolumn{1}{r}{$0.2260$} &
\multicolumn{1}{r}{} & \multicolumn{1}{r}{$0.0807$} &
\multicolumn{1}{r}{$2.4378$} & \multicolumn{1}{r}{$0.0946$} &
\multicolumn{1}{r}{$0.614$}\\
$2$ & \multicolumn{1}{r}{$-0.0020$} & \multicolumn{1}{r}{$3.6881$} &
\multicolumn{1}{r}{${0.0608}$} & \multicolumn{1}{r}{$0.9400$} &
\multicolumn{1}{r}{} & \multicolumn{1}{r}{$-0.0151$} &
\multicolumn{1}{r}{$2.6022$} & \multicolumn{1}{r}{${0.0532}$} &
\multicolumn{1}{r}{$0.936$}\\
$3$ & \multicolumn{1}{r}{$-0.0493$} & \multicolumn{1}{r}{$3.8578$} &
\multicolumn{1}{r}{$0.0793$} & \multicolumn{1}{r}{$0.8530$} &
\multicolumn{1}{r}{} & \multicolumn{1}{r}{$-0.0299$} &
\multicolumn{1}{r}{$2.7221$} & \multicolumn{1}{r}{$0.0601$} &
\multicolumn{1}{r}{$0.902$}\\
$4$ & \multicolumn{1}{r}{$-0.0556$} & \multicolumn{1}{r}{$4.0160$} &
\multicolumn{1}{r}{$0.0843$} & \multicolumn{1}{r}{$0.8300$} &
\multicolumn{1}{r}{} & \multicolumn{1}{r}{$-0.0259$} &
\multicolumn{1}{r}{$2.8112$} & \multicolumn{1}{r}{$0.0590$} &
\multicolumn{1}{r}{$0.912$}\\
$5$ & \multicolumn{1}{r}{$-0.0494$} & \multicolumn{1}{r}{$4.1629$} &
\multicolumn{1}{r}{$0.0813$} & \multicolumn{1}{r}{$0.8590$} &
\multicolumn{1}{r}{} & \multicolumn{1}{r}{$-0.0198$} &
\multicolumn{1}{r}{$2.8760$} & \multicolumn{1}{r}{$0.0572$} &
\multicolumn{1}{r}{$0.924$}\\
$6$ & \multicolumn{1}{r}{$-0.0409$} & \multicolumn{1}{r}{$4.2930$} &
\multicolumn{1}{r}{$0.0773$} & \multicolumn{1}{r}{$0.8830$} &
\multicolumn{1}{r}{} & \multicolumn{1}{r}{$-0.0151$} &
\multicolumn{1}{r}{$2.9209$} & \multicolumn{1}{r}{$0.0561$} &
\multicolumn{1}{r}{$0.937$}\\
$7$ & \multicolumn{1}{r}{$-0.0333$} & \multicolumn{1}{r}{$4.4024$} &
\multicolumn{1}{r}{$0.0742$} & \multicolumn{1}{r}{$0.9080$} &
\multicolumn{1}{r}{} & \multicolumn{1}{r}{$-0.0120$} &
\multicolumn{1}{r}{$2.9512$} & \multicolumn{1}{r}{$0.0556$} &
\multicolumn{1}{r}{$0.946$}\\
$8$ & \multicolumn{1}{r}{$-0.0272$} & \multicolumn{1}{r}{$4.4908$} &
\multicolumn{1}{r}{$0.0723$} & \multicolumn{1}{r}{$0.9190$} &
\multicolumn{1}{r}{} & \multicolumn{1}{r}{$-0.0100$} &
\multicolumn{1}{r}{$2.9712$} & \multicolumn{1}{r}{$0.0554$} &
\multicolumn{1}{r}{$0.952$}\\
$9$ & \multicolumn{1}{r}{$-0.0225$} & \multicolumn{1}{r}{$4.5604$} &
\multicolumn{1}{r}{$0.0712$} & \multicolumn{1}{r}{$0.9270$} &
\multicolumn{1}{r}{} & \multicolumn{1}{r}{$-0.0087$} &
\multicolumn{1}{r}{$2.9845$} & \multicolumn{1}{r}{$0.0553$} &
\multicolumn{1}{r}{$0.953$}\\
$10$ & \multicolumn{1}{r}{$-0.0191$} & \multicolumn{1}{r}{$4.6142$} &
\multicolumn{1}{r}{$0.0706$} & \multicolumn{1}{r}{$0.9340$} &
\multicolumn{1}{r}{} & \multicolumn{1}{r}{$-0.0078$} &
\multicolumn{1}{r}{$2.9934$} & \multicolumn{1}{r}{$0.0553$} &
\multicolumn{1}{r}{${0.951}$}\\
$20$ & \multicolumn{1}{r}{$-0.0096$} & \multicolumn{1}{r}{$4.7858$} &
\multicolumn{1}{r}{$0.0698$} & \multicolumn{1}{r}{$0.9380$} &
\multicolumn{1}{r}{} & \multicolumn{1}{r}{$-0.0046$} &
\multicolumn{1}{r}{$3.0148$} & \multicolumn{1}{r}{$0.0551$} &
\multicolumn{1}{r}{$0.955$}\\
$40$ & \multicolumn{1}{r}{$-0.0075$} & \multicolumn{1}{r}{$4.8264$} &
\multicolumn{1}{r}{$0.0699$} & \multicolumn{1}{r}{$0.9390$} &
\multicolumn{1}{r}{} & \multicolumn{1}{r}{$-0.0018$} &
\multicolumn{1}{r}{$3.0159$} & \multicolumn{1}{r}{$0.0549$} &
\multicolumn{1}{r}{$0.956$}\\
$60$ & \multicolumn{1}{r}{$-0.0062$} & \multicolumn{1}{r}{$4.8383$} &
\multicolumn{1}{r}{$0.0698$} & \multicolumn{1}{r}{$0.9370$} &
\multicolumn{1}{r}{} & \multicolumn{1}{r}{$-0.0007$} &
\multicolumn{1}{r}{$3.0150$} & \multicolumn{1}{r}{$0.0549$} &
\multicolumn{1}{r}{$0.957$}\\
$80$ & \multicolumn{1}{r}{$-0.0051$} & \multicolumn{1}{r}{$4.8446$} &
\multicolumn{1}{r}{$0.0698$} & \multicolumn{1}{r}{$0.9380$} &
\multicolumn{1}{r}{} & \multicolumn{1}{r}{$-0.0003$} &
\multicolumn{1}{r}{$3.0146$} & \multicolumn{1}{r}{$0.0549$} &
\multicolumn{1}{r}{$0.957$}\\
$100$ & \multicolumn{1}{r}{$-0.0043$} & \multicolumn{1}{r}{$4.8489$} &
\multicolumn{1}{r}{$0.0698$} & \multicolumn{1}{r}{$0.9380$} &
\multicolumn{1}{r}{} & \multicolumn{1}{r}{${-0.0001}$} &
\multicolumn{1}{r}{$3.0145$} & \multicolumn{1}{r}{$0.0549$} &
\multicolumn{1}{r}{$0.958$}\\
$200$ & \multicolumn{1}{r}{$-0.0018$} & \multicolumn{1}{r}{$4.8588$} &
\multicolumn{1}{r}{$0.0697$} & \multicolumn{1}{r}{$0.9420$} &
\multicolumn{1}{r}{} & \multicolumn{1}{r}{$0.0003$} &
\multicolumn{1}{r}{$3.0150$} & \multicolumn{1}{r}{$0.0549$} &
\multicolumn{1}{r}{$0.957$}\\
$400$ & \multicolumn{1}{r}{$-0.0005$} & \multicolumn{1}{r}{$4.8603$} &
\multicolumn{1}{r}{$0.0697$} & \multicolumn{1}{r}{$0.9440$} &
\multicolumn{1}{r}{} & \multicolumn{1}{r}{$0.0006$} &
\multicolumn{1}{r}{$3.0151$} & \multicolumn{1}{r}{$0.0549$} &
\multicolumn{1}{r}{$0.957$}\\
$600$ & \multicolumn{1}{r}{$-0.0003$} & \multicolumn{1}{r}{$4.8611$} &
\multicolumn{1}{r}{$0.0697$} & \multicolumn{1}{r}{$0.9440$} &
\multicolumn{1}{r}{} & \multicolumn{1}{r}{$0.0008$} &
\multicolumn{1}{r}{$3.0152$} & \multicolumn{1}{r}{$0.0549$} &
\multicolumn{1}{r}{$0.957$}\\
$800$ & \multicolumn{1}{r}{$-0.0003$} & \multicolumn{1}{r}{$4.8630$} &
\multicolumn{1}{r}{$0.0697$} & \multicolumn{1}{r}{${0.9450}$} &
\multicolumn{1}{r}{} & \multicolumn{1}{r}{$0.0009$} &
\multicolumn{1}{r}{$3.0153$} & \multicolumn{1}{r}{$0.0549$} &
\multicolumn{1}{r}{$0.957$}\\
$1000$ & \multicolumn{1}{r}{${-0.0002}$} & \multicolumn{1}{r}{$4.8652$} &
\multicolumn{1}{r}{$0.0698$} & \multicolumn{1}{r}{${0.9450}$} &
\multicolumn{1}{r}{} & \multicolumn{1}{r}{$0.0010$} &
\multicolumn{1}{r}{$3.0156$} & \multicolumn{1}{r}{$0.0549$} &
\multicolumn{1}{r}{$0.959$}\\\hline\hline
\end{tabular}
\vspace{-0.2cm}
\caption{\label{SimulationsExample}
Simulation results for $n=1000$. Setup as in Table~\ref{BiasAvarExample}. Results based on 1000 Monte Carlo replications.}
\end{center}
\end{table}

\begin{table}[p]
\begin{center}
\begin{tabular}
[c]{rrrrrrrrrr}\hline\hline
$q$ & $\theta_{\ast}-\theta_{0}$ & \multicolumn{1}{c}{$V_{\ast}$} & $\mathrm{RMSE}$ & $\mathrm{CI}_{0.95}$ &  & $\theta_{\ast}-\theta_{0}$ & \multicolumn{1}{c}{$V_{\ast}$} & $\mathrm{RMSE}$ & $\mathrm{CI}_{0.95}$\\\hline
& \multicolumn{4}{c}{$T_{0}=1,T_{1}=3,T=4$} &  & \multicolumn{4}{c}{$T_{0}=1,T_{1}=9,T=10$}\\\cline{2-5}\cline{7-10}
$0$ & $0.6704$ & ${2.7135}$ & $0.6725$ & $0.0000$ &  & $0.6721$ & ${1.4919}$ & $0.6732$ & $0.0000$\\
$1$ & $0.3131$ & $3.0434$ & $0.3179$ & $0.0001$ &  & $0.2545$ & $2.1031$ & $0.2586$ & $0.0002$\\
$2$ & $0.0758$ & $3.6051$ & $0.0967$ & $0.7564$ &  & $0.0567$ & $2.7400$ & $0.0771$ & $0.8087$\\
$3$ & $-0.0252$ & $3.9220$ & ${0.0675}$ & $0.9312$ &  & $0.0026$ & $2.9696$ & ${0.0546}$ & $0.9497$\\
$4$ & $-0.0576$ & $4.0648$ & $0.0859$ & $0.8525$ &  & $-0.0122$ & $3.0664$ & $0.0567$ & $0.9444$\\
$5$ & $-0.0641$ & $4.1438$ & $0.0908$ & $0.8310$ &  & $-0.0172$ & $3.1265$ & $0.0585$ & $0.9391$\\
$6$ & $-0.0623$ & $4.2037$ & $0.0899$ & $0.8395$ &  & $-0.0192$ & $3.1701$ & $0.0595$ & $0.9366$\\
$7$ & $-0.0583$ & $4.2566$ & $0.0875$ & $0.8545$ &  & $-0.0200$ & $3.2030$ & $0.0600$ & $0.9356$\\
$8$ & $-0.0544$ & $4.3047$ & $0.0852$ & $0.8683$ &  & $-0.0202$ & $3.2285$ & $0.0603$ & $0.9354$\\
$9$ & $-0.0510$ & $4.3484$ & $0.0833$ & $0.8793$ &  & $-0.0201$ & $3.2485$ & $0.0604$ & $0.9356$\\
$10$ & $-0.0481$ & $4.3877$ & $0.0819$ & $0.8877$ &  & $-0.0199$ & $3.2646$ & $0.0605$ & $0.9360$\\
$20$ & $-0.0354$ & $4.6294$ & $0.0767$ & $0.9184$ &  & $-0.0164$ & $3.3403$ & $0.0601$ & $0.9408$\\
$40$ & $-0.0281$ & $4.8096$ & $0.0748$ & $0.9310$ &  & $-0.0130$ & $3.3925$ & $0.0597$ & $0.9442$\\
$60$ & $-0.0250$ & $4.8776$ & $0.0742$ & $0.9351$ &  & $-0.0114$ & $3.4177$ & $0.0596$ & $0.9456$\\
$80$ & $-0.0231$ & $4.9177$ & $0.0738$ & $0.9375$ &  & $-0.0104$ & $3.4340$ & $0.0595$ & $0.9464$\\
$100$ & $-0.0217$ & $4.9491$ & $0.0736$ & $0.9391$ &  & $-0.0097$ & $3.4462$ & $0.0595$ & $0.9469$\\
$200$ & $-0.0173$ & $5.0518$ & $0.0732$ & $0.9432$ &  & $-0.0078$ & $3.4839$ & $0.0595$ & $0.9480$\\
$400$ & $-0.0147$ & $5.1243$ & $0.0731$ & $0.9452$ &  & $-0.0065$ & $3.5246$ & $0.0597$ & $0.9486$\\
$600$ & $-0.0141$ & $5.1505$ & $0.0731$ & $0.9455$ &  & $-0.0058$ & $3.5524$ & $0.0599$ & $0.9489$\\
$800$ & $-0.0139$ & $5.1726$ & $0.0732$ & $0.9457$ &  & $-0.0054$ & $3.5744$ & $0.0600$ & $0.9491$\\
$1000$ & $-0.0137$ & $5.1968$ & $0.0734$ & $0.9459$ &  & $-0.0051$ & $3.5935$ & $0.0602$ & $0.9492$\\
$\infty$ & ${-0.0}_{3}{78}$ & $15.7761$ & $0.1256$ & ${0.9500}$ &  & ${-0.0}_{6}{20}$ & $35.4401$ & $0.1883$ & ${0.9500}$\\\hline\hline
\end{tabular}
\vspace{-0.2cm}
\caption{\label{BiasAvarExample2}
Asymptotic biases and variances, and approximate RMSEs and coverage rates
($n=1000$). The model is as in Example~\ref{ex:StaticPanel2}
with standard normal errors, $\pi_{\mathrm{prior}}(\alpha)=\phi
(\alpha)$, $\pi_{0}(\alpha)=\phi(\alpha-1)$, and $\theta_{0}=1$. The pseudo-true
value $\theta_{\ast}$ is based on the integrated score.
Notation: $0_{r}=\underset{r\text{ zeros}}{\underbrace{00\ldots0}}$. 
}
\end{center}
\end{table}

\begin{table}[p]
\begin{center}
\begin{tabular}
[c]{rrrrrrrrrr}\hline\hline
$q$ & bias & $n\times$var & RMSE & ${\rm CI}_{0.95}$ &  & bias & $n\times$var &
${\rm RMSE}$ & ${\rm CI}_{0.95}$\\\hline
& \multicolumn{4}{c}{$T_{0}=1,T_{1}=3,T=4$} &  & \multicolumn{4}{c}{$T_{0}=1,%
T_{1}=9,T=10$}\\\cline{2-5}\cline{7-10}%
$0$ & $0.6721$ & ${2.6261}$ & $0.6740$ & $0.0000$ &  & $0.6740$ &
${1.4847}$ & $0.6751$ & \multicolumn{1}{r}{$0.0000$}\\
$1$ & $0.3147$ & $3.1535$ & $0.3197$ & $0.0000$ &  & $0.2559$ & $2.1953$ &
$0.2602$ & \multicolumn{1}{r}{$0.0000$}\\
$2$ & $0.0774$ & $3.8340$ & $0.0991$ & $0.7500$ &  & $0.0576$ & $2.8768$ &
$0.0787$ & \multicolumn{1}{r}{$0.7950$}\\
$3$ & $-0.0239$ & $4.1970$ & $0.0690$ & $0.9160$ &  & $0.0034$ & $3.1192$ &
${0.0560}$ & \multicolumn{1}{r}{$0.9470$}\\
$4$ & $-0.0562$ & $4.3578$ & $0.0867$ & $0.8480$ &  & $-0.0115$ & $3.2218$ &
$0.0579$ & \multicolumn{1}{r}{$0.9420$}\\
$5$ & $-0.0627$ & $4.4455$ & $0.0915$ & $0.8320$ &  & $-0.0166$ & $3.2857$ &
$0.0597$ & \multicolumn{1}{r}{$0.9360$}\\
$6$ & $-0.0608$ & $4.5111$ & $0.0906$ & $0.8360$ &  & $-0.0186$ & $3.3324$ &
$0.0606$ & \multicolumn{1}{r}{$0.9330$}\\
$7$ & $-0.0568$ & $4.5686$ & $0.0883$ & $0.8550$ &  & $-0.0194$ & $3.3679$ &
$0.0612$ & \multicolumn{1}{r}{$0.9310$}\\
$8$ & $-0.0528$ & $4.6209$ & $0.0861$ & $0.8630$ &  & $-0.0196$ & $3.3955$ &
$0.0615$ & \multicolumn{1}{r}{$0.9310$}\\
$9$ & $-0.0493$ & $4.6682$ & $0.0843$ & $0.8720$ &  & $-0.0195$ & $3.4172$ &
$0.0616$ & \multicolumn{1}{r}{$0.9320$}\\
$10$ & $-0.0464$ & $4.7110$ & $0.0829$ & $0.8810$ &  & $-0.0192$ & $3.4347$ &
$0.0617$ & \multicolumn{1}{r}{$0.9320$}\\
$20$ & $-0.0334$ & $4.9732$ & $0.0780$ & $0.9050$ &  & $-0.0157$ & $3.5184$ &
$0.0613$ & \multicolumn{1}{r}{$0.9370$}\\
$40$ & $-0.0258$ & $5.1642$ & $0.0764$ & $0.9130$ &  & $-0.0122$ & $3.5780$ &
$0.0610$ & \multicolumn{1}{r}{$0.9410$}\\
$60$ & $-0.0226$ & $5.2313$ & $0.0758$ & $0.9190$ &  & $-0.0105$ & $3.6080$ &
$0.0610$ & \multicolumn{1}{r}{$0.9420$}\\
$80$ & $-0.0206$ & $5.2686$ & $0.0754$ & $0.9210$ &  & $-0.0094$ & $3.6283$ &
$0.0610$ & \multicolumn{1}{r}{$0.9430$}\\
$100$ & $-0.0190$ & $5.2971$ & $0.0752$ & $0.9260$ &  & $-0.0087$ & $3.6438$ &
$0.0610$ & \multicolumn{1}{r}{$0.9430$}\\
$200$ & $-0.0145$ & $5.3904$ & $0.0748$ & ${0.9350}$ &  & $-0.0066$ &
$3.6935$ & $0.0611$ & \multicolumn{1}{r}{$0.9420$}\\
$400$ & $-0.0117$ & $5.4577$ & ${0.0748}$ & $0.9330$ &  & $-0.0051$ &
$3.7479$ & $0.0614$ & \multicolumn{1}{r}{$0.9430$}\\
$600$ & $-0.0111$ & $5.4847$ & $0.0749$ & ${0.9350}$ &  & $-0.0042$ &
$3.7836$ & $0.0617$ & \multicolumn{1}{r}{$0.9440$}\\
$800$ & $-0.0108$ & $5.5093$ & $0.0750$ & ${0.9350}$ &  & $-0.0037$ &
$3.8109$ & $0.0618$ & \multicolumn{1}{r}{${0.9480}$}\\
$1000$ & ${-0.0106}$ & $5.5370$ & $0.0752$ & $0.9340$ &  &
${-0.0033}$ & $3.8338$ & $0.0620$ & \multicolumn{1}{r}{$0.9470$%
}\\\hline\hline
\end{tabular}
\vspace{-0.2cm}
\caption{\label{SimulationsExample2}
Simulation results for $n=1000$. Setup as in Table~\ref{BiasAvarExample2}. Results based on 1000 Monte Carlo replications.}
\end{center}
\end{table}

\bigskip

\begin{table}[p]
\begin{center}
\begin{tabular}
[c]{rrrrrrrrrr}\hline\hline
$q$ & $\theta_{\ast}-\theta_{0}$ & \multicolumn{1}{c}{$V_{\ast}$} & $\mathrm{RMSE}$ & $\mathrm{CI}_{0.95}$ &  & $\theta_{\ast}-\theta_{0}$ & \multicolumn{1}{c}{$V_{\ast}$} & $\mathrm{RMSE}$ & $\mathrm{CI}_{0.95}$\\\hline
& \multicolumn{4}{c}{$T=4$} &  & \multicolumn{4}{c}{$T=6$}\\\cline{2-5}\cline{7-10}
$0$ & $0.6218$ & $4.3158$ & $0.6252$ & $0.0000$ & \multicolumn{1}{r}{} & $0.4940$ & $2.8654$ & $0.4969$ & $0.0000$\\
$1$ & $0.2514$ & $4.4379$ & $0.2601$ & $0.0348$ & \multicolumn{1}{r}{} & $0.1278$ & $2.8694$ & $0.1386$ & $0.3351$\\
$2$ & $0.0498$ & $4.7217$ & $0.0849$ & $0.8880$ & \multicolumn{1}{r}{} & $0.0022$ & $3.0428$ & $0.0552$ & $0.9498$\\
$3$ & $-0.0244$ & $4.9279$ & $0.0743$ & $0.9360$ & \multicolumn{1}{r}{} & $-0.0247$ & $3.1636$ & $0.0614$ & $0.9276$\\
$4$ & $-0.0434$ & $5.0908$ & $0.0835$ & $0.9066$ & \multicolumn{1}{r}{} & $-0.0257$ & $3.2518$ & $0.0626$ & $0.9263$\\
$5$ & $-0.0435$ & $5.2366$ & $0.0844$ & $0.9076$ & \multicolumn{1}{r}{} & $-0.0221$ & $3.3188$ & $0.0617$ & $0.9330$\\
$6$ & $-0.0385$ & $5.3672$ & $0.0828$ & $0.9178$ & \multicolumn{1}{r}{} & $-0.0187$ & $3.3687$ & $0.0610$ & $0.9380$\\
$7$ & $-0.0330$ & $5.4795$ & $0.0811$ & $0.9269$ & \multicolumn{1}{r}{} & $-0.0164$ & $3.4052$ & $0.0606$ & $0.9409$\\
$8$ & $-0.0284$ & $5.5726$ & $0.0799$ & $0.9333$ & \multicolumn{1}{r}{} & $-0.0148$ & $3.4321$ & $0.0604$ & $0.9427$\\
$9$ & $-0.0247$ & $5.6478$ & $0.0791$ & $0.9375$ & \multicolumn{1}{r}{} & $-0.0138$ & $3.4523$ & $0.0603$ & $0.9437$\\
$10$ & $-0.0220$ & $5.7077$ & $0.0787$ & $0.9402$ & \multicolumn{1}{r}{} & $-0.0131$ & $3.4679$ & $0.0603$ & $0.9443$\\
$20$ & $-0.0150$ & $5.9317$ & $0.0785$ & $0.9457$ & \multicolumn{1}{r}{} & $-0.0099$ & $3.5369$ & $0.0603$ & $0.9468$\\
$40$ & $-0.0136$ & $6.0276$ & $0.0788$ & $0.9465$ & \multicolumn{1}{r}{} & $-0.0066$ & $3.5811$ & $0.0602$ & $0.9486$\\
$60$ & $-0.0123$ & $6.0717$ & $0.0789$ & $0.9472$ & \multicolumn{1}{r}{} & $-0.0052$ & $3.5984$ & $0.0602$ & $0.9491$\\
$80$ & $-0.0110$ & $6.1029$ & $0.0789$ & $0.9477$ & \multicolumn{1}{r}{} & $-0.0046$ & $3.6070$ & $0.0602$ & $0.9493$\\
$100$ & $-0.0100$ & $6.1281$ & $0.0789$ & $0.9481$ & \multicolumn{1}{r}{} & $-0.0042$ & $3.6126$ & $0.0603$ & $0.9494$\\
$200$ & $-0.0069$ & $6.2028$ & $0.0791$ & $0.9491$ & \multicolumn{1}{r}{} & $-0.0036$ & $3.6283$ & $0.0603$ & $0.9496$\\
$400$ & $-0.0053$ & $6.2445$ & $0.0792$ & $0.9495$ & \multicolumn{1}{r}{} & $-0.0031$ & $3.6447$ & $0.0605$ & $0.9497$\\
$600$ & $-0.0050$ & $6.2586$ & $0.0793$ & $0.9495$ & \multicolumn{1}{r}{} & $-0.0028$ & $3.6564$ & $0.0605$ & $0.9498$\\
$800$ & $-0.0049$ & $6.2713$ & $0.0793$ & $0.9496$ & \multicolumn{1}{r}{} & $-0.0026$ & $3.6670$ & $0.0606$ & $0.9498$\\
$1000$ & $-0.0048$ & $6.2849$ & $0.0794$ & $0.9496$ & \multicolumn{1}{r}{} & $-0.0024$ & $3.6765$ & $0.0607$ & $0.9498$\\\hline\hline
\end{tabular}
\vspace{-0.2cm}
\caption{\label{ContinuousRegressor}
Asymptotic biases and variances, and approximate RMSEs and coverage rates
($n=1000$). The model is as in Example~\ref{ex:StaticPanel}
with standard normal errors, $X_{it}\sim\mathcal{N}(0.5,0.25)$, $\pi_{\mathrm{prior}}(\alpha)=\phi
(\alpha)$, $\pi_{0}(\alpha)=\phi(\alpha-1)$, and $\theta_{0}=1$. The pseudo-true
value $\theta_{\ast}$ is based on the integrated score.
}
\end{center}
\end{table}

\bigskip

\begin{table}[p]
\begin{center}
\begin{tabular}
[c]{rrrrrrrrrr}\hline\hline
$q$ & bias & $n\times$var & RMSE & ${\rm CI}_{0.95}$ &  & bias & $n\times$var & ${\rm RMSE}$ & ${\rm CI}_{0.95}$\\\hline
& \multicolumn{4}{c}{$T=4$} &  & \multicolumn{4}{c}{$T=6$}\\\cline{2-5}\cline{7-10}
$0$ & $0.6237$ & $4.3199$ & $0.6272$ & $0.0000$ & \multicolumn{1}{r}{} & $0.4936$ & $2.6107$ & $0.4962$ & $0.0000$\\
$1$ & $0.2524$ & $4.6806$ & $0.2615$ & $0.0230$ & \multicolumn{1}{r}{} & $0.1262$ & $2.6395$ & $0.1363$ & $0.3390$\\
$2$ & $0.0502$ & $5.1142$ & $0.0874$ & $0.8830$ & \multicolumn{1}{r}{} & $0.0002$ & $2.8155$ & $0.0531$ & $0.9660$\\
$3$ & $-0.0241$ & $5.3716$ & $0.0772$ & $0.9210$ & \multicolumn{1}{r}{} & $-0.0268$ & $2.9237$ & $0.0603$ & $0.9270$\\
$4$ & $-0.0430$ & $5.5593$ & $0.0861$ & $0.8910$ & \multicolumn{1}{r}{} & $-0.0278$ & $2.9999$ & $0.0614$ & $0.9270$\\
$5$ & $-0.0429$ & $5.7236$ & $0.0870$ & $0.8940$ & \multicolumn{1}{r}{} & $-0.0242$ & $3.0578$ & $0.0603$ & $0.9310$\\
$6$ & $-0.0377$ & $5.8700$ & $0.0854$ & $0.9040$ & \multicolumn{1}{r}{} & $-0.0208$ & $3.1013$ & $0.0594$ & $0.9360$\\
$7$ & $-0.0321$ & $5.9960$ & $0.0838$ & $0.9130$ & \multicolumn{1}{r}{} & $-0.0184$ & $3.1334$ & $0.0589$ & $0.9390$\\
$8$ & $-0.0272$ & $6.1008$ & $0.0827$ & $0.9190$ & \multicolumn{1}{r}{} & $-0.0168$ & $3.1572$ & $0.0586$ & $0.9420$\\
$9$ & $-0.0235$ & $6.1859$ & $0.0821$ & $0.9280$ & \multicolumn{1}{r}{} & $-0.0158$ & $3.1751$ & $0.0585$ & $0.9490$\\
$10$ & $-0.0207$ & $6.2543$ & $0.0817$ & $0.9280$ & \multicolumn{1}{r}{} & $-0.0150$ & $3.1891$ & $0.0584$ & $0.9500$\\
$20$ & $-0.0132$ & $6.5187$ & $0.0818$ & $0.9380$ & \multicolumn{1}{r}{} & $-0.0118$ & $3.2520$ & $0.0582$ & $0.9510$\\
$40$ & $-0.0117$ & $6.6324$ & $0.0823$ & $0.9400$ & \multicolumn{1}{r}{} & $-0.0084$ & $3.2934$ & $0.0580$ & $0.9550$\\
$60$ & $-0.0102$ & $6.6792$ & $0.0824$ & $0.9400$ & \multicolumn{1}{r}{} & $-0.0070$ & $3.3097$ & $0.0580$ & $0.9550$\\
$80$ & $-0.0089$ & $6.7104$ & $0.0824$ & $0.9390$ & \multicolumn{1}{r}{} & $-0.0063$ & $3.3178$ & $0.0579$ & $0.9560$\\
$100$ & $-0.0078$ & $6.7350$ & $0.0824$ & $0.9410$ & \multicolumn{1}{r}{} & $-0.0060$ & $3.3229$ & $0.0580$ & $0.9560$\\
$200$ & $-0.0046$ & $6.8060$ & $0.0826$ & $0.9450$ & \multicolumn{1}{r}{} & $-0.0054$ & $3.3360$ & $0.0580$ & $0.9570$\\
$400$ & $-0.0029$ & $6.8440$ & $0.0828$ & $0.9440$ & \multicolumn{1}{r}{} & $-0.0049$ & $3.3463$ & $0.0581$ & $0.9570$\\
$600$ & $-0.0025$ & $6.8576$ & $0.0828$ & $0.9440$ & \multicolumn{1}{r}{} & $-0.0045$ & $3.3526$ & $0.0581$ & $0.9580$\\
$800$ & $-0.0024$ & $6.8704$ & $0.0829$ & $0.9450$ & \multicolumn{1}{r}{} & $-0.0043$ & $3.3587$ & $0.0581$ & $0.9590$\\
$1000$ & $-0.0023$ & $6.8843$ & $0.0830$ & $0.9450$ & \multicolumn{1}{r}{} & $-0.0041$ & $3.3647$ & $0.0581$ & $0.9580$\\\hline\hline
\end{tabular}
\vspace{-0.2cm}
\caption{\label{ContinuousRegressorSimulations}
Simulation results for $n=1000$. Setup as in Table~\ref{ContinuousRegressor}. Results based on 1000 Monte Carlo replications.
}
\end{center}
\end{table}

\begin{table}[p]
\begin{center}
\begin{tabular}{ccrrrrrrrr}
\hline\hline
   $\pi_0$  & $q$ &    \multicolumn{1}{c}{$T=64$} &    \multicolumn{1}{c}{$T=128$} &
\multicolumn{1}{c}{$T=256$} & \multicolumn{1}{c}{$T=512$} & & \multicolumn{1}{c}{$T=128$} &
\multicolumn{1}{c}{$T=256$} & \multicolumn{1}{c}{$T=512$} \\
\multicolumn{2}{c}{} & \multicolumn{4}{c}{$b_T$} &  & \multicolumn{3}{c}{$b_{T/2}/b_T$}\\
\cline{1-6}\cline{8-10}
  $\delta_0$   &0&   $0.0219$ & $0.0110$ & $0.0055$ & $0.0028$ & & $1.99$ & $2.00$ & $2.00$ \\ 
      &1&   $0.0_{3}94$ & $0.0_{3}23$ & $0.0_{4}57$ & $0.0_{4}14$ & & $4.09$ & $4.05$ & $4.03$ \\ 
      &2&   $-0.0_{4}14$ & $-0.0_{5}27$ & $-0.0_{6}39$ & $-0.0_{7}52$ & & $5.15$ & $6.88$ & $7.50$ \\ 
      &3&   $-0.0_{5}94$ & $-0.0_{6}49$ & $-0.0_{7}28$ & $-0.0_{8}16$ & & $19.30$ & $17.69$ & $16.98$ \\ \hline
  $\delta_1$   &0&   $0.1157$ & $0.0594$ & $0.0301$ & $0.0151$ & & $1.95$ & $1.98$ & $1.99$ \\ 
      &1&   $0.0051$ & $0.0014$ & $0.0_{3}35$ & $0.0_{4}89$ & & $3.75$ & $3.90$ & $3.96$ \\ 
      &2&   $0.0_{3}90$ & $0.0_{4}91$ & $0.0_{5}93$ & $0.0_{5}11$ & & $9.87$ & $9.85$ & $8.52$ \\ 
      &3&   $0.0_{3}13$ & $-0.0_{3}24$ & $-0.0_{5}13$ & $-0.0_{7}41$ & & $-5.58$ & $18.39$ & $31.73$ \\ \hline
  $\delta_2$   &0&   $0.4635$ & $0.27621$ & $0.1540$ & $0.0819$ & & $1.68$ & $1.79$ & $1.88$ \\ 
      &1&   $-0.0521$ & $-0.0137$ & $-0.0025$ & $-0.0_{3}40$ & & $3.81$ & $5.40$ & $6.37$ \\ 
      &2&   $-0.0218$ & $0.0_{3}30$ & $0.0012$ & $0.0_{3}29$ & & $-73.60$ & $0.24$ & $4.20$ \\ 
      &3&   $-0.0045$ & $0.0037$ & $0.0012$ & $0.0_{3}10$ & & $-1.24$ & $3.05$ & $11.70$ \\ \hline
  $U[0.5,1.5]$ 
      &0&    $0.1065$ & $0.0548$ & $0.0278$ & $0.0140$ & & $1.94$ & $1.97$ & $1.99$ \\ 
      &1&   $0.0043$ & $0.0012$ & $0.0_{3}31$ & $0.0_{4}79$ & & $3.61$ & $3.84$ & $3.93$ \\ 
      &2&   $0.0_{3}87$ & $0.0_{3}13$ & $0.0_{4}13$ & $0.0_{5}13$ & & $6.91$ & $9.60$ & $9.76$ \\ 
      &3&   $0.0_{3}34$ & $0.0_{5}48$ & $-0.0_{5}28$ & $-0.0_{6}20$ & & $70.53$ & $-1.68$ & $14.56$ \\ \hline
 $U[0,2]$ 
      &0&   $0.0863$ & $0.0446$ & $0.0227$ & $0.0114$ & & $1.94$ & $1.97$ & $1.98$ \\ 
      &1&   $0.0022$ & $0.0_{3}68$ & $0.0_{3}20$ & $0.0_{4}53$ & & $3.27$ & $3.47$ & $3.70$ \\ 
      &2&   $0.0_{3}37$ & $0.0_{4}14$ & $0.0_{4}30$ & $0.0_{5}42$ & & $2.66$ & $4.66$ & $7.10$ \\ 
      &3&   $0.0_{3}33$ & $0.0_{4}89$ & $0.0_{5}99$ & $-0.0_{8}88$ & & $3.73$ & $8.97$ & $-1130.78$ \\ \hline\hline
\end{tabular}
\vspace{-0.2cm}
\caption{\label{Rates}
Asymptotic bias rates. The model is as in Example~\ref{ex:StaticPanel}
with standard normal errors, $T_0=T_1=T/2$, $\pi_{\mathrm{prior}}(\alpha)=\phi
(\alpha)$, and $\pi_{0}$ as given in the table. The left part of the table gives $b_T=\theta_*-\theta_0$ for given $T$ and $q$. The three rightmost columns give the ratio $b_{T/2}/b_T$.
Notation: $\delta_z$ is the distribution with all mass at $z$; 
$0_{r}=\underset{r\text{ zeros}}{\underbrace{00\ldots0}}$.
}
\end{center}    
\end{table}

\newpage
\section{Some further remarks and ideas}
\label{sec:Extensions}

\subsection{Alternative bias correction methods}
\label{subsec:AlternativeBC}

Let $\widehat \alpha(y,x,\theta) :=  \arg\max_{\alpha\in{\cal A}}f\left(y \, \big| \, x, \alpha, \theta \right)$
be the MLE of $\alpha$ for fixed $\theta$. Define
$ \widetilde  Q(\widetilde y \, | \, y,x,\theta)
  := f\left(\widetilde y \, \big| \, x, \widehat \alpha(y,x,\theta), \theta \right)$,
and let
 $ \widetilde  Q(x,\theta) $ be the 
$n_{\cal Y} \times n_{\cal Y}$ matrix with elements $ \widetilde  Q_{k,\ell}(x,\theta)  =  \widetilde  Q(y_{(k)} \, | \, y_{(\ell)},x,\theta)$,
for $k,\ell \in \{1,\ldots,n_{\cal Y}\}$.

Instead of implementing the bias correction of the score as in \eqref{S1Formula}, one could alternatively consider
\begin{align}
    \widetilde s^{(1)}(y,x,\theta) &:= s(y,x,\theta) -  \sum_{\widetilde y \in {\cal Y}}  s(\widetilde y,x,\theta)   \,  f\left(\widetilde y \, \big| \, x, \widehat \alpha(y,x,\theta) , \theta \right)   
    \nonumber  \\
      &= s(y,x,\theta) -  \sum_{\widetilde y \in {\cal Y}}  s(\widetilde y,x,\theta) \, \widetilde Q(\widetilde y \, | \, y,x,\theta) 
     \nonumber \\
      &=   S(x,\theta) \, \left[ \mathbbm{I}_{n_{\cal Y}} -  \widetilde  Q(x,\theta) \right] \, \delta(y) \, .
       \label{S1Formula2}
\end{align}
This alternative bias correction method is very natural: We simply have subtracted from the original score the expression for the bias in the first line of
\eqref{InfeasibleBiasCorrection}, and replaced the unknown $A$ with the estimator $\widehat \alpha(y,x,\theta)$.
In fact, this is exactly the ``profile-score adjustment'' to the score function that is suggested in \cite{dhaene2015profile}.

The expression in \eqref{S1Formula2} is identical to that in \eqref{S1Formula}, except that $Q(x,\theta)$
is replaced with $\widetilde  Q(x,\theta)$. Iterating this alternative bias correction $q$ times therefore also gives   the formula in
\eqref{GeneralIteration} with $Q(x,\theta)$ replaced with $\widetilde  Q(x,\theta)$. 
Thus, by the same arguments as before, for large values of $q$, 
the corresponding score function $\widetilde 
 s^{(q)}(x,\theta)$ will be dominated by contributions from eigenvectors of $\widetilde  Q(x,\theta)$
that correspond to eigenvalues close to or equal to zero.

It is therefore  natural to ask why in our presentation above we have chosen the bias correction in  \eqref{S1Formula} based 
on the posterior distribution of $A$ instead of the bias correction in \eqref{S1Formula2} based on the MLE of $A$.
The answer is that the matrix $\widetilde  Q(x,\theta)$ does not have the same convenient algebraic properties as the matrix 
$Q(x,\theta)$. In particular, none of the parts (i), (ii), (iii) of Lemma~\ref{ExactFD} would hold if we replaced $Q(x,\theta)$ by $\widetilde  Q(x,\theta)$,
implying that the close relationship between the bias correction in \eqref{S1Formula2} and functional differencing does not generally hold for the alternative bias correction discussed here. 

To explain why $\widetilde  Q(x,\theta)$ does not have these properties, consider the following.
For given values of $x$ and $\theta$, assume that there exist two outcomes $y$ and $\bar y$
that give the same MLE of $A$, that is, 
 $\widehat \alpha(y,x,\theta)  = \widehat \alpha(\bar y,x,\theta)$. Then,
 the two columns of  $\widetilde  Q(x,\theta)$ that correspond to $y$ and $\bar y$ are identical, and therefore 
 $\widetilde  Q(x,\theta)$ 
 does not have full rank, implying that it has
a zero eigenvalue.
 The existence of this zero eigenvalue is simply a consequence of $\widehat \alpha(y,x,\theta)  = \widehat \alpha(\bar y,x,\theta)$.
 
 Now, in models where there exists a sufficient statistic for $A$ (conditional on $X$), if the outcomes $y$ and $\bar y$ have the same value of the sufficient statistic, then $\widehat \alpha(y,x,\theta)  = \widehat \alpha(\bar y,x,\theta)$, 
 and in that case the zero eigenvalue of $\widetilde  Q(x,\theta)$  just discussed
 is closely related to functional differencing because the existence of the sufficient statistic generates 
 valid moment functions; recall the example in equation \eqref{MomentFctSufficient}.
 
 However, we may also have  $\widehat \alpha(y,x,\theta)  = \widehat \alpha(\bar y,x,\theta)$  for reasons
 that have nothing to do with functional differencing. For example, consider Example~\ref{ex:StaticPanel2}
 with normally distributed errors, $T \geq 2$ even, and $T_0=T_1=T/2$.
 Then,
all outcomes $y$ with $y_0 + y_1 = T/2$ have $\widehat \alpha(y,\theta)=-\theta/2$. 
So there are at least $1+T/2$ different outcomes with the same value of 
$\widehat \alpha(y,\theta)$, which 
 implies that  $\widetilde  Q(\theta)$ has at least $T/2$ zero eigenvalues. But we have found numerically that  $Q(\theta)$ does not have
zero eigenvalues in this example for $\theta \neq 0$, that is, according to Lemma~\ref{ExactFD}, no exact moment function exists.

This example shows that $Q(x,\theta)$ and $\widetilde  Q(x,\theta)$ have different 
algebraic properties, and 
it explains why we have focused on $Q(x,\theta)$ instead of $\widetilde  Q(x,\theta)$ in our discussion.
Nevertheless, the observation that bias correction can be iterated using
  the formula in \eqref{GeneralIteration} can be useful for alternative methods as well.

\subsection{Alternative ways to implement approximate functional differencing} 
\label{subsec:AlternativeImplementations}

Instead of choosing $s^{(q)}(y,x,\theta)$  as moment function to estimate $\theta_0$, we could alternatively
choose the moment function $s_h(y,x,\theta)$ defined 
by \eqref{DefhQ} and \eqref{GeneralScoreFunctions}
for some other function $h:[0,1] \rightarrow \mathbb{R}$.
In particular, one very natural relaxation of 
$s_{\infty}(y,x,\theta)$ 
and  $h_\infty(\lambda) =    \mathbbm{1}\{  \lambda = 0 \}$
would be to choose
 \begin{align*}
   h(\lambda)  &= K(\lambda/c) \, ,
\end{align*}   
for some soft-thresholding function
$K : [0,\infty) \rightarrow [0,\infty)$,
for example, $K(\xi)=\exp(-\xi)$.
The tuning parameter $q \in \{0,1,2,\ldots\}$   is replaced here
by the bandwidth parameter $c>0$, which specifies which
eigenvalues of $Q(\theta,x)$ are considered to be  close to zero. 
Regarding the thresholding function, one could 
in principle consider a simple indicator function
$K(\xi)=\mathbbm{1}\{ \xi\leq 1 \}$, but 
since this function is discontinuous,
the resulting score function $s_{h}(y,x,\theta)$ defined
in \eqref{GeneralScoreFunctions} would then
be discontinuous in $\theta$, so we would not recommend
this.

Another possibility to implement approximate functional differencing is to replace  the set ${\cal A}$ by
a finite set ${\cal A}_*$ with cardinality $n_{\cal A}$ less than
$n_{\cal Y}$. As explained in Remark 2 above, after this replacement,
the matrix $Q(x,\theta)$ will have at least $n_{\cal Y}- n_{\cal A}$ zero eigenvalues,
that is, one can then use the moment function
$s_{\infty}(y,x,\theta)$ defined in \eqref{SCOREinf}
to implement the MM or GMM estimator. In this case, the
key tuning parameter to choose is the number of points $n_{\cal A}$ in the set ${\cal A}_*$ that ``approximates'' ${\cal A}$.

\subsection{Average effect estimation}
\label{sec:AverageEffects}

In models of the form \eqref{model2} we are often not only
interested in the unknown $\theta_0$ but also in 
functionals of the unknown $\pi_0(\alpha|x)$. 
In particular, consider average effects of the form
\begin{align*}
   \mu_0 &= \mathbbm{E}\left[ \mu(X,A,\theta_0) \right]
    = \mathbbm{E}\left[ \int_{\cal A}   \mu(X, \alpha , \theta_0)  \,   \pi_0(\alpha\,|\,X) \, {\rm d}\alpha \right] ,
\end{align*}
where $\mu( x,\alpha , \theta)$ is a known function that 
specifies the average effect of interest. 
For example, in a panel data model,
if we are interested in the average partial effect with respect to 
the $p$-th regressor in period $t$, we could choose
$\mu( x,\alpha , \theta) = \frac{\partial}{\partial x_{t,p}} \sum_{y \in {\cal Y}} y_t f(y | x, \theta,\alpha)  $.
For other examples of functionals of the individual-specific effects, see e.g.\ \cite{arellano2012identifying}.

We now focus on the problem of
estimating $\mu_0$. Therefore, in this subsection, we assume that
the problem of estimating $\theta_0$ is already resolved (with corresponding estimator $\widehat \theta)$,
and we focus on the problem that 
$\pi_0(\alpha|x)$ is unknown when estimating average effects $\mu_0$.

Analogously to the iterated bias-corrected score functions
$s^{(q)}(y,x,\theta) $ in \eqref{GeneralIteration},
we want to define a sequence of estimating functions
$w^{(q)}(y,x,\theta) $, $q=0,1,2,\ldots$, such that, for some $q$,
$$
    \mu_*^{(q)} :=  \mathbbm{E}\left[ w^{(q)}(Y,X,\theta_0) \right]
$$
is close to $\mu_0 $. The corresponding estimator of $\mu_0$ is
$$
   \widehat \mu^{(q)} := \frac 1 n \sum_{i=1}^n 
   w^{(q)} ( Y_i,X_i,\widehat \theta) \, .
$$
Using the posterior distribution in
\eqref{DefPost}, a natural baseline estimating
function ($q=0$) is
\begin{align*}
     w^{(0)}(y,x,\theta) := 
   \int_{\cal A}   \mu( x, \alpha , \theta)  \,    \pi_{\rm post}(\alpha \,|\, y,x,\theta)  \, {\rm d}\alpha \, .
\end{align*}  
The corresponding estimator $\widehat \mu^{(0)}$ of $\mu_0$ 
can again be motivated by ``large-$T$'' panel data considerations,
where, under regularity conditions, the 
posterior distribution concentrates around the true value $A$
as $T \rightarrow \infty$.

Let $W(x,\theta)$  be the  $n_{\cal Y}$-vector with entries $w^{(0)}(y_{(k)},x,\theta)$,
$k=1,\ldots,n_{\cal Y}$. Then, the analog of the 
limiting estimating function in \eqref{SCOREinf},
corresponding to $q \rightarrow \infty$, for average effects is
\begin{align}
     w^{(\infty)}(y,x,\theta) &:= 
   W'(x,\theta) \,   Q^\dagger(x,\theta)  \, \delta(y)
   \nonumber \\
   &= W'(x,\theta) \;   \widetilde h^{(\infty)}[ Q(x,\theta) ]  \; \delta(y) \, ,
   &
     \widetilde h^{(\infty)}(\lambda) 
    &:= \left\{
    \begin{array}{ll}
    \lambda^{-1}
    & \text{for } \lambda>0 ,
    \\
    0
     & \text{for } \lambda=0 ,
    \end{array} \right.
   \label{Winfty}
\end{align}
where $Q^\dagger(x,\theta) $ is a pseudo-inverse of $Q(x,\theta) $, and the application of a function $ \widetilde h^{(\infty)}:[0,1] \rightarrow \mathbb{R}$
to the matrix $Q(x,\theta)$ was defined
in equation \eqref{DefhQ}.
The motivation
for choosing $w^{(\infty)}(y,x,\theta)$ in this way is that
it gives an unbiased estimator of the average effect
(i.e., $\mu_*^{(\infty)} = \mu_0$) whenever we can write
$\mu( x, \alpha , \theta) =
\sum_{y \in {\cal Y}} 
 \nu(y,x,\theta)
 f\left(y \, \big| \, x, \alpha, \theta \right)$
 for some function $\nu(y,x,\theta)$.\footnote{
 This is because in that special case we have $W'(x,\theta)=N'(x,\theta) \, Q(x,\theta)$,
where $N(x,\theta)$  is the  $n_{\cal Y}$-vector with entries $\nu(y_{(k)},x,\theta)$, and therefore
$  \mathbb{E}\left[  w^{(\infty)}(Y,X,\theta_0)  \, \big| \, X=x, \, A = \alpha  \right]  =  N'(x,\theta_0) \, Q(x,\theta_0) \,   Q^\dagger(x,\theta_0)   \mathbb{E}\left[ \delta(Y)   \, \big| \, X=x, \, A = \alpha  \right] =  N'(x,\theta_0) \,   \mathbb{E}\left[ \delta(Y)   \, \big| \, X=x, \, A = \alpha  \right] = \mu( x, \alpha , \theta_0)$.}
Of course, average effects with this form
of $\mu( \alpha , x, \theta)$ are a very special case, but
they are usually the only cases for which we can
expect unbiased estimation of the average effect to be feasible
(for fixed $T$); see also
\cite{aguirregabiria2021identification}.
Notice that we do not assume here that $\mu( \alpha , x, \theta)$ is of this form,
it is just used to motivate \eqref{Winfty}.
 
As we have seen before, the non-zero eigenvalues 
of $Q(x,\theta) $ can be very small, which implies that the pseudo-inverse
$Q^\dagger(x,\theta) $ can have very large elements. The
corresponding estimator $\widehat \mu^{(\infty)}$ based on \eqref{Winfty} therefore
typically has a very large variance and we do not recommend
this estimator in practice. Instead,
to balance the bias-variance trade-off 
of the average effect estimator, some
regularization of the pseudo-inverse of $Q(x,\theta) $
in \eqref{Winfty} is required. 
There are various ways to implement
regularization, in the same way that there
are various ways to implement
approximate functional differencing
(see Section~\ref{subsec:AlternativeImplementations}).

Here, regularization means that
we want to find
functions $\widetilde h_q(\lambda)$ that
approximate the inverse function $1/\lambda$
well for large values of $\lambda \in [0,1]$, but that deviate from 
$1/\lambda$ for values of $\lambda$
close to zero to avoid divergence.\footnote{
In previous sections, the functions $h_q(\lambda) =  (1-\lambda)^q$ were polynomial approximations
of (rescaled versions of) the function
$h_\infty(\lambda)=\mathbbm{1}\{  \lambda = 0 \}$.
The regularization that is analogous
to $ s^{(q)}(y,x,\theta)$ in \eqref{GeneralIteration}
is given by a $q$-th order Taylor expansion of
the function $1/\lambda$ around $\lambda=1$.}
This gives,\footnote{
Here, we use the convention that $0^0=1$,
which also implies that $\left[\mathbbm{I}_{n_{\cal Y}} - Q(x,\theta)
     \right]^0 = \mathbbm{I}_{n_{\cal Y}}$ even though $Q(x,\theta)$ has an eigenvalue equal to one.
 Also, there is some ambiguity in 
 what value we should assign to $\widetilde h_q(\lambda) $ for $\lambda=0$. 
 We choose  $\widetilde h_q(0)=q+1$ because
 it results in the simple polynomial expression
 \eqref{WqGeneral} for  $\widetilde h_q[ Q(x,\theta) ]$, which is convenient since
 $\widetilde h_q[ Q(x,\theta) ]$ can be evaluated
 without ever calculating the eigenvalues
 and eigenvectors of $Q(x,\theta)$.
 However, if we want to obtain
 $w^{(\infty)}(y,x,\theta) $ in \eqref{Winfty} as
 the limit of $w^{(q)}(y,x,\theta)$
 as $q \rightarrow \infty$, then we should
 assign $\widetilde h_q(0)=0$ for $\lambda=0$, but this would deviate
 from the polynomial expression.
}
for $q\in\{0,1,2,\ldots\}$,
\begin{align}
  \widetilde h_q(\lambda)
   &=  
   \sum_{r=0}^q (1-\lambda)^r
   =\left\{
    \begin{array}{ll}
    \frac{ 1 - (1-\lambda)^{q+1} } {\lambda}
    & \text{for } \lambda>0 ,
    \\
    q+1 
     & \text{for } \lambda=0 .
    \end{array} \right.  
   \label{TildeHpolynomial}
\end{align}
The corresponding estimating function
that regularizes $w^{(\infty)}(y,x,\theta)$ is 
therefore given by
\begin{align}
     w^{(q)}(y,x,\theta) &:= 
   W'(x,\theta) \;   \widetilde h_q[ Q(x,\theta) ]  \; \delta(y),
   &
     \widetilde h_q[ Q(x,\theta) ]
 = \sum_{r=0}^q  \left[\mathbbm{I}_{n_{\cal Y}} - Q(x,\theta)
     \right]^r  \, .
   \label{WqGeneral}
\end{align}
This is a polynomial in $ Q(x,\theta)$,
as  was the case for
$ s^{(q)}(y,x,\theta)$.
Choosing a value of $q$ that is not too large
  therefore ensures that the variance of
 the corresponding estimator  $\widehat \mu^{(q)}$
 remains reasonably small (for fixed $q$), because we don't
 need the pseudo-inverse of $Q(x,\theta)$.

 Note also that $ w^{(q)}(y,x,\theta) $
 and the corresponding estimators  $\widehat \mu^{(q)}$
 have a large-$T$ bias-correction interpretation
 very similar to $ s^{(q)}(y,x,\theta) $. For example,
we have $\widetilde h_1(\lambda)=2-\lambda$,
and therefore 
 $$ w^{(1)}(y,x,\theta) = 2 \, w^{(0)}(y,x,\theta)
  - 
   W'(x,\theta) \;   Q(x,\theta)  \; \delta(y) \, .
$$   
We conjecture that
the estimator of $\mu_0$ corresponding to only
$ W'(x,\theta) Q(x,\theta)  \delta(y) $
has twice the leading order $1/T$ asymptotic
bias of the estimator  $\widehat \mu^{(0)}$ 
corresponding to $w^{(0)}(y,x,\theta)$, that is,
$ w^{(1)}(y,x,\theta) $ is exactly the jackknife
linear combination that eliminates the large-$T$ leading
order bias in $\widehat \mu^{(0)}$; see \cite{DhaeneJochmans2015}.
Appropriate iterations of this jackknife bias
correction also give the estimating functions  $w^{(q)}(y,x,\theta)$ for $q>1$.
 
We are not considering average effects further here.
But we found it noteworthy that 
there is a formalism for average effect
calculation that closely mirrors the development
of approximate functional differencing for the 
estimation of $\theta_0$ introduced above.
However, this does not imply that we expect the results
for average-effect estimation to be necessarily similar to 
those for the estimation of the common parameters $\theta_0$.
In particular, for small values of $T$, the identified set for the average effects in discrete-choice panel data  models tends to be much larger  than the identified set of the common parameters
(see, e.g.,\ \citealt{chernozhukov2013average}, \citealt{davezies2021identification}, \citealt{LiuPoirierShiu21}, and \citealt{pakel2021bounds}). Therefore we expect larger values of $T$
to be required for the point estimators $\widehat \mu^{(q)}$
to perform well, and we also expect the bias-variance trade-off 
in the choice of $q$ to be quite different.
For a closely related discussion
see \cite{bonhomme2017panel},
and also the section on ``Average marginal effects'' in the
2010 working paper version
of \cite{bonhomme2012functional}.

\section{Conclusions}
\label{sec:conclusions}

We have linked 
the large-$T$ panel data literature with the functional differencing
method through a bias correction that converges
to functional differencing when iterated.
Our numerical
illustrations show that in models where exact functional 
differencing is not possible, one may still apply
it approximately to obtain estimates that can be essentially 
unbiased, even when the number of time periods $T$ is small.

The key element in our construction
is the $n_Y \times n_Y$ matrix
$Q(x,\theta)$. 
The eigenvalues of this matrix 
are informative about whether (approximate) functional 
differencing is applicable in a given model.
The matrix $Q(x,\theta)$ also features prominently in our 
bias-corrected score functions in \eqref{GeneralIteration}
and in our regularized estimating functions for
average effects in \eqref{WqGeneral}.
We have assumed a discrete outcome space with a finite number of elements $n_Y$. When the outcome space is infinite, the matrix 
$Q(x,\theta)$ has to be replaced by  the
corresponding operator.

The goal of this paper was primarily to introduce 
and illustrate an approximate version of functional differencing.
Future work is needed to better understand the properties 
of the method and to explore its usefulness in empirical work,
both for the estimation of common parameters,
which was our primary focus, 
and for the estimation of average effects, briefly introduced in Section~\ref{sec:AverageEffects}.

%\bibliographystyle{chicago3}
%\bibliography{refs}

\appendix

\section{Proofs}

\begin{proof}[\bf Proof of Lemma~\ref{lemma:MatrixQ}]
Define
\begin{align*}
    \overline Q(\widetilde y \, | \, y,x,\theta) =  \frac{\int_{\cal A} f \left(\widetilde y \, \big| \, x,\alpha,\theta \right)  f(y \,  | \,x,\alpha,\theta) \, \pi_{\rm prior}(\alpha\,|\,x)  {\rm d} \alpha}
    {\left[ p_{\rm prior}(\widetilde y \, |\, x,\theta) \right]^{1/2}
    \left[ p_{\rm prior}( y \, |\, x,\theta) \right]^{1/2} } 
\end{align*}
and let $\overline Q(x,\theta)$ be the $n_{\cal Y} \times n_{\cal Y}$ matrix with elements $\overline 
 Q_{k,\ell}(x,\theta)  = \overline  Q(y_{(k)} \, | \, y_{(\ell)},x,\theta)$. 
Also define the  $n_{\cal Y} \times n_{\cal Y}$ diagonal matrix
$$
   P_{\rm prior}( x,\theta)  = {\rm diag}\left[   p_{\rm prior}( y_{(k)} \, |\, x,\theta) \right]_{k=1,\ldots,n_{\cal Y}}   .
$$
From \eqref{DefPost} and \eqref{DefQ} we obtain
\begin{align*}
    Q( \widetilde y \, | \, y,x,\theta) 
    &=  \frac{\int_{\cal A} f \left( \widetilde y \, \big| \, x,\alpha,\theta \right)  f( y \,  | \,x,\alpha,\theta) \, \pi_{\rm prior}(\alpha\,|\,x)  {\rm d} \alpha} {p_{\rm prior}(y \, |\, x,\theta)}  
    \\
    &= 
    \left[ p_{\rm prior}( \widetilde y \, |\, x,\theta) \right]^{1/2}
  \,  \overline Q( \widetilde y \, | \, y,x,\theta) \,
     \left[ p_{\rm prior}( y \, |\, x,\theta) \right]^{-1/2} ,
\end{align*}
which in matrix notation is
\begin{align*}
    Q(x,\theta) 
      &= 
    \left[ P_{\rm prior}(  x,\theta) \right]^{1/2}
  \,  \overline Q(x,\theta) \,
     \left[ P_{\rm prior}(  x,\theta) \right]^{-1/2} .
\end{align*}
This shows that the matrices $Q(x,\theta)$ and $ \overline Q(x,\theta)$
are similar and therefore have the same eigenvalues.\footnote{%
          Two matrices $A$ and $B$ are similar if $B=P^{-1} A P$ for some nonsingular matrix $P$. Similar matrices have the same eigenvalues.
          }
The matrix        $ \overline Q(x,\theta)$ is symmetric and positive
semi-definite (by construction), which implies that all its eigenvalues
(and therefore all eigenvalues of $Q(x,\theta)$) are non-negative real numbers.
Furthermore,  $ \overline Q(x,\theta)$  is diagonalizable because it is
symmetric. Hence $  Q(x,\theta)$  is also diagonalizable,
because it is similar to $ \overline Q(x,\theta)$.\footnote{
 A matrix  is diagonalizable if and only if it is similar to a diagonal matrix. 
 Since  $ \overline Q(x,\theta)$ is similar to a diagonal matrix, and 
  $  Q(x,\theta)$  is  similar to $ \overline Q(x,\theta)$, it must also be
  the case that $  Q(x,\theta)$  is similar to a diagonal matrix.
}

 In addition,  $Q(x,\theta)$ is a stochastic matrix (by construction), 
 which implies that its spectral radius is equal to one, that is,
                   $Q(x,\theta)$ cannot have any eigenvalue larger than one.
                   We thus conclude that all eigenvalues of $Q(x,\theta)$  lie in the interval $[0,1]$.
\end{proof}

The following lemma is useful for the proof of Lemma~\ref{ExactFD},
which we present afterward.

\begin{lemma}
   \label{lemma:help}
    Let the assumptions of Lemma~\ref{ExactFD}
    hold. Let $w(y,x,\theta_0) \in \mathbb{R}$
    be such that
    $$ 
    \sum_{y \in {\cal Y}}
       w(y,x,\theta_0) \,Q\left( y \, \big| \,\widetilde y, x,\theta_0 \right)  = 0 \, , \qquad
       \textrm{for all } \widetilde y \in {\cal Y} \, .
    $$
    Then
    $$ 
       \sum_{y \in {\cal Y}}
       w(y,x,\theta_0) \,f \left(y \, \big| \, x,\alpha,\theta_0 \right)  = 0 \, ,
       \qquad\textrm{for all } \alpha \in {\cal A} \, .
    $$
\end{lemma}

\begin{proof}[\bf Proof of Lemma~\ref{lemma:help}]
    The $n_{\cal Y} \times n_{\cal Y}$ diagonal matrix $P_{\rm prior}( x,\theta_0)$
    was defined in the Proof of Lemma~\ref{lemma:MatrixQ}.
   In addition, 
    let $F(x,\alpha,\theta_0)$ and $W(x,\theta_0)$ be the $n_{\cal Y}$-vectors with elements 
$f(y_{(k)} \, | \, x,\alpha,\theta_0)$ and $w(y_{(k)},x,\theta_0)$, respectively,
for $k=1,\ldots,n_{\cal Y}$.
     Then
     \begin{align}
        Q(x,\theta_0) = 
        \int_{\cal A} F(x,\alpha,\theta_0) \, F'(x,\alpha,\theta_0) 
        \, \pi_{\rm prior}(\alpha\,|\,x) \, {\rm d}\alpha \, P^{-1}_{\rm prior}( x,\theta_0) \, ,
        \label{QmatrixMat}
     \end{align}
    and the condition on  $w(y,x,\theta_0)$ in the lemma can be written as
    $$
        W'(x,\theta_0) \, Q(x,\theta_0) = 0 \, .
    $$
    Plugging in the expression for
     $Q(x,\theta_0)$ in \eqref{QmatrixMat} and
    multiplying   with $P_{\rm prior}( x,\theta_0) \, W(x,\theta_0) $
    from the right gives
    $$
        \int_{\cal A} W'(x,\theta_0) \, F(x,\alpha,\theta_0) \, F'(x,\alpha,\theta_0) 
         \, W(x,\theta_0) \, \, \pi_{\rm prior}(\alpha\,|\,x) \, {\rm d}\alpha  = 0 \, .
    $$
    Since   $W'(x,\theta_0) \, F(x,\alpha,\theta_0) \, F'(x,\alpha,\theta_0) 
         \, W(x,\theta_0) \geq 0$
   and $\pi_{\rm prior}(\alpha\,|\,x)>0$  we conclude that
   \begin{align}
     W'(x,\theta_0) \, F(x,\alpha,\theta_0) \, F'(x,\alpha,\theta_0) 
         \, W(x,\theta_0) = 0 \, ,   \label{WFFW}
   \end{align}
for almost all values $\alpha$, except possibly for a set of values $\alpha$ that has measure
   zero under $\pi_{\rm prior}(\alpha\,|\,x)$. However,  
   since $f\left(y \, \big| \, x, \alpha, \theta_0 \right)$
   is assumed to be continuous in $\alpha$, we conclude that \eqref{WFFW} must 
   hold for all $\alpha \in {\cal A}$, since any violation on a set of measure
   zero would require a discontinuity in $\alpha$. 
   Finally, \eqref{WFFW} also implies that 
   $$
      W'(x,\theta_0) \, F(x,\alpha,\theta_0) = 0 \, ,
   $$
   for all $\alpha \in {\cal A}$. This is what we wanted to show, just written in vector notation.
\end{proof}

\begin{proof}[\bf Proof of Lemma~\ref{ExactFD}]
   \underline{\# part (i):}
    Let $U_0(x,\theta_0) $ be the submatrix of $U(x,\theta_0)$
    that only contains those columns that are the right-eigenvectors of $Q(x,\theta_0)$ corresponding
    to the eigenvalues $\lambda_j(x,\theta_0) = 0$.
    We then have 
    $
       Q(x,\theta_0) \, U_0(x,\theta_0)  = 0.
    $
    Similarly, let $[U^{-1}(x,\theta_0)]_0$ be the submatrix
    of $U^{-1}(x,\theta_0)$ that only contains the  
    rows that are the left-eigenvectors of $Q(x,\theta_0)$ corresponding
    to the eigenvalues $\lambda_j(x,\theta_0) = 0$.
    We then have
    $$
          [U^{-1}(x,\theta_0)]_0 \, Q(x,\theta_0)  = 0 \, ,
    $$
    and according to Lemma~\ref{lemma:help} this implies
   \begin{align}
        [U^{-1}(x,\theta_0)]_0 \, F(x,\alpha,\theta_0) &= 0 \, .
        \label{VinvProperty}
   \end{align}
    Next, by using the definition of $ s_\infty(y,x,\theta_0)$
    and $ h_\infty[ Q(x,\theta_0) ]$ in the main text we find
   \begin{align*}
 s_\infty(y,x,\theta_0) &=  S(x,\theta_0)  \,  h_\infty[ Q(x,\theta_0) ]   \, \delta(y) 
 \\
  &=  S(x,\theta_0)  \, U(x,\theta_0)  \;   \diag \left( \big[ \mathbbm{1}\left\{  \lambda_j(x,\theta_0) = 0 \right\}  \big]_{j=1,\ldots,n_{\cal Y}} \right) \; U^{-1}(x,\theta_0)    \, \delta(y) 
 \\ 
  &=  S(x,\theta_0)  \, U_0(x,\theta_0)  \;    [U^{-1}(x,\theta_0)]_0    \, \delta(y) \, ,
\end{align*}
and therefore
\begin{align*}
    \mathbb{E}\left[  s_{\infty}(Y,X,\theta_0)  \, \big| \, X=x, \, A = \alpha  \right] &= S(x,\theta_0)  \, U_0(x,\theta_0)  \;    [U^{-1}(x,\theta_0)]_0    \, F(x,\alpha,\theta_0) 
     = 0 \, ,
\end{align*}
where in the last step we used \eqref{VinvProperty}.

\medskip
 
\underline{\# part (ii):}
Let  ${\mathfrak m}(x,\theta_0)$ be the $n_{\cal Y}$-vector with elements 
${\mathfrak m}(y_{(k)} , x, \theta_0)$,
$k=1,\ldots,n_{\cal Y}$.
Then, 
$\mathbb{E}\left[ {\mathfrak m}(Y , X, \theta_0) \, \big| \, X=x, \, A = \alpha  \right] = 0$
can be written in vector notation as
\begin{align}
   {\mathfrak m}'(x,\theta_0) \,  F(x,\alpha,\theta_0)
   &= 0 \, .
   \label{MomentVector}
\end{align}
From the expression of $ Q(x,\theta_0)$ in
\eqref{QmatrixMat} we see that this implies
${\mathfrak m}'(x,\theta_0) Q(x,\theta_0) = 0$,
that is, if \eqref{MomentVector} holds for
all $\alpha \in {\cal A}$, then $ Q(x,\theta_0)$
has a zero eigenvalue with corresponding left-eigenvector
${\mathfrak m}(x,\theta_0)$. This is the ``if'' part of the 
statement in part (ii) of the lemma.

Conversely, if $ Q(x,\theta_0)$ has a zero eigenvalue,
then let ${\mathfrak m}(x,\theta_0)$  be a corresponding
left-eigenvector. We then have ${\mathfrak m}'(x,\theta_0) Q(x,\theta_0) = 0$. According to Lemma~\ref{lemma:help}
this implies that \eqref{MomentVector} holds, or equivalently
that $\mathbb{E}\left[ {\mathfrak m}(Y , X, \theta_0) \, \big| \, X=x, \, A = \alpha  \right] = 0$. We have thus also shown the
``only if'' part of the 
statement in part (ii) of the lemma.

\medskip

\underline{\# part (iii):}
Let  ${\mathfrak m}(y , x, \theta_0) \in \mathbb{R}$  be such that
$\mathbb{E}\left[ {\mathfrak m}(Y , X, \theta_0) \, \big| \, X=x, \, A = \alpha  \right] = 0$. We choose 
$$
   s(y , x, \theta_0)  =  {\mathfrak m}(y , x, \theta_0) \, .
$$
Using the definition of $ s^{(1)}(y,x,\theta)$ in
\eqref{S1Formula} we then find
$s^{(1)}(y,x,\theta_0) =  {\mathfrak m}(y , x, \theta_0)$,
and therefore also
$$
  s^{(q)}(y,x,\theta_0) = {\mathfrak m}(y , x, \theta_0) \, ,
$$
for all $q \in \{1,2,\ldots\}$. We therefore also find
$s_\infty(y,x,\theta_0) = \lim_{q \rightarrow \infty} s^{(q)}(y,x,\theta_0) = {\mathfrak m}(y , x, \theta_0)$,
which is what we wanted to show.   
\end{proof}

\end{document}